\documentclass[journal,final,twocolumn,10pt,twoside]{IEEEtran}

\usepackage[utf8]{inputenc} 
\usepackage[T1]{fontenc}
\usepackage{url}
\usepackage{ifthen}
\usepackage{cite}
\usepackage{amsmath}%
\usepackage{wasysym}%
\usepackage{amssymb}

\usepackage{mdwmath}
\usepackage{blindtext}
\usepackage{eqparbox}
\DeclareMathOperator{\E}{\mathbb{E}}

\usepackage[english]{babel}
\usepackage{lipsum}
\usepackage{xcolor}
\usepackage{tikz}
\usepackage{graphicx}
\usepackage{algorithm, algorithmic}
\usepackage{siunitx}
\definecolor{blue2}{RGB}{0,100,200}
\usepackage{calc}
\newcolumntype{M}[1]{>{\centering\arraybackslash}m{#1}}
\newcolumntype{N}{@{}m{0pt}@{}}

\usepackage[a-1b]{pdfx}
\usepackage{hyperref}

\usepackage{etoolbox}
\AtBeginEnvironment{tabular}{\tiny}

\usepackage{amsthm}

\usepackage{calc}
\newcolumntype{M}[1]{>{\centering\arraybackslash}m{#1}}
\newcolumntype{N}{@{}m{0pt}@{}}

\newtheorem{theorem}{{Theorem}}
\newtheorem{lemma}[theorem]{{Lemma}}
\newtheorem{proposition}[theorem]{{Proposition}}
\newtheorem{corollary}[theorem]{{Corollary}}

\newcommand{\cG}{{\cal G}}

\newcommand{\cR}{{\cal R}}

\DeclareMathAlphabet{\mathbfsl}{OT1}{ppl}{b}{it} 

\newcommand{\bv}{\mathbfsl{v}}

\newcommand{\br}{\mathbfsl{r}}

\newcommand{\be}[1]{\begin{equation}\label{#1}}
\newcommand{\ee}{\end{equation}} 
\newcommand{\eq}[1]{(\ref{#1})}

\renewcommand{\le}{\leqslant} 
\renewcommand{\leq}{\leqslant}
 
\renewcommand{\geq}{\geqslant}

\renewcommand{\Bbb}{\mathbb}

\newcommand{\R}{{\Bbb R}}

\newcommand{\Tref}[1]{Theo\-rem\,\ref{#1}}
\newcommand{\Pref}[1]{Pro\-po\-si\-tion\,\ref{#1}}
\newcommand{\Lref}[1]{Lem\-ma\,\ref{#1}}
\newcommand{\Cref}[1]{Co\-ro\-lla\-ry\,\ref{#1}}

\newcommand{\norm}[1]{\left\lVert#1\right\rVert}

\newcommand{\Fq}{{{\Bbb F}}_{\!q}}

\IEEEoverridecommandlockouts
\begin{document}
\title{Coded Computing via Binary Linear Codes:\\ Designs and Performance Limits} 
\author{Mahdi Soleymani$^*$, Mohammad Vahid Jamali$^*$, and \\Hessam Mahdavifar,~\IEEEmembership{Member,~IEEE}
		\thanks{$^*$Authors of equal contribution.}
		\thanks{The material in this paper was presented in part at the IEEE Information Theory Workshop (ITW), Visby, Sweden, Aug. 2019 \cite{jamali2019coded}.}
		\thanks{The authors are with the Department of Electrical Engineering and Computer Science, University of Michigan, Ann Arbor, MI 48109, USA (e-mail: mahdy@umich.edu, mvjamali@umich.edu, hessam@umich.edu).}
		\thanks{This work was supported by the National Science Foundation under grants CCF--1763348, CCF--1909771, and  CCF--1941633.}
	}
	
\maketitle
\begin{abstract}
We consider the problem of coded distributed computing where a large linear computational job, such as a matrix multiplication, is divided into $k$ smaller tasks, encoded using an $(n,k)$ linear code, and performed over $n$ distributed nodes. The goal is to reduce the average execution time of the computational job. We provide a connection between the problem of characterizing the average execution time of a coded distributed computing system and the problem of analyzing the error probability of codes of length $n$ used over erasure channels. Accordingly, we present closed-form expressions for the execution time using binary random linear codes and the best execution time any linear-coded distributed computing system can achieve. It is also shown that there exist \textit{good} binary linear codes that not only attain (asymptotically) the best performance that any linear code (not necessarily binary) can achieve but also are numerically stable against the inevitable rounding errors in practice.
We then develop a low-complexity algorithm for decoding Reed-Muller (RM) codes over erasure channels. Our decoder only involves additions, subtractions, {and inversion of relatively small matrices of dimensions at most $\log n+1$}, and enables coded computation over real-valued data. Extensive numerical analysis of the fundamental results as well as RM- and polar-coded computing schemes demonstrate the excellence of the RM-coded computation in achieving close-to-optimal performance while having a low-complexity decoding and explicit construction.
 The proposed framework in this paper enables efficient designs of distributed computing systems given the rich literature in the channel coding theory. 
\end{abstract}

\section{Introduction}
There has been an increasing interest in recent years toward applying ideas from coding theory to improve the performance of various computation, communication, and networking applications. For example, ideas from repetition coding has been applied to several setups in computer networks, e.g., by running a request over multiple servers and waiting for the first completion of the request by discarding the rest of the request duplicates \cite{ananthanarayanan2013effective,vulimiri2013low,gardner2015reducing}. Another direction is to investigate the application of coding theory in cloud networks and distributing computing systems \cite{jonas2017occupy,lee2018speeding}.
{In general, coding
techniques can be applied to improve the run-time performance of distributed computing systems.}

Distributed computing refers to the problem of performing a large computational job over many, say $n$, nodes with limited processing capabilities. A coded computing scheme aims to divide the job to $k<n$ tasks and then to introduce $n-k$ redundant tasks using an $(n,k)$ code, in order to alleviate the effect of slower nodes, also referred to as \textit{stragglers}. In such a setup, it is often assumed that each node is assigned one task and hence, the total number of \textit{encoded tasks} is $n$ equal to the number of nodes. 

Recently, there has been extensive research activities to
leverage coding schemes in order to boost the performance of distributed computing systems \cite{li2020coded,lee2018speeding,li2016unified,reisizadeh2017coded,li2017coding,yang2017computing,lee2017high,dutta2016short,wang2018coded,dutta2019optimal,yu2020straggler,prakash2020coded, jahani2018codedsketch}. 
However, most of the work in the literature focus on the application of maximum distance separable (MDS) codes. This is while encoding and decoding of MDS codes over real numbers, especially when the number of servers is large, e.g., more than $100$, face several barriers, such as numerical stability {issues} and decoding complexity. In particular, decoding of MDS codes is not robust against unavoidable rounding errors when used over real numbers
\cite{higham2002accuracy}. Quantizing the real-valued data and mapping them to a finite field over which the computations are carried out  
\cite{yu2019lagrange} can be an alternative approach. However, performing computations over  finite fields imposes further numerical barriers due to overflow errors when used over real-valued data.

As we will show in Section III, MDS codes are \textit{theoretically} optimal in terms of minimizing the average execution time of any linear-coded distributed computing system. However, as discussed above,
their application comes with some practical impediments,
either when used over real-valued inputs or large finite fields, in most of coded computing applications comprised of large number of local nodes. A sub-optimal yet practically interesting approach is to apply binary linear codes, 
{with generator matrices consisting of $-1$'s  and $1$'s}, 
and then perform the computation over real numbers. In this case, there is no need for the quantization as  the encoded tasks sent to the worker nodes are  obtained from a linear combination of the uncoded tasks merely involving additions and subtractions.  
Inspired by this, in this paper, we consider $(n,k)$ binary linear codes where all computations are performed over real-valued data inputs. To this end, we first derive several fundamental limits to characterize the performance of coded computing schemes employing binary linear codes. We then investigate Reed-Muller (RM) coded computation enabled by our proposed low-complexity algorithm for decoding RM codes over erasure channels. Our decoding algorithm is specifically designed to work over real-valued data and only involves additions, subtractions, {and inversion of relatively small matrices of dimensions at most $1+ \log n$}.

\subsection{Related Work}\label{sec:related}
Coded computing paradigm provides a framework to address critical issues that arise in large-scale distributed computing and learning problems such as stragglers, security, and privacy by combining coding theory and distributed computing. 
For instance, in \cite{lee2018speeding,li2016unified,yu2020straggler,reisizadeh2019coded,li2016coded,li2017coding,yang2017computing,lee2017high},  coding theoretic techniques have been utilized to combat the deteriorating effects of stragglers in coded computing schemes. The adaptation of such protocols to the analog domain often results in numerical instability, i.e., the accuracy in the computation outcome drops significantly when the number of servers grows large. Furthermore, coded computing schemes have been proposed that enable data privacy and security \cite{yu2018lagrange,yu2020entangled,aliasgari2020private,d2020gasp,bitar2019private,nodehi2019secure,soleymani2021list}. However, in these prior works the data is first quantized and then mapped to a finite field where the tools from coding theory over finite fields can be applied. The performance of such schemes also drops sharply when the dataset size passes a certain threshold due to overflow errors \cite{soleymani2020privacy, soleymani2020analog}.        

There is another line of work concerning the adoption of coded computing schemes for straggler mitigation in the analog domain  \cite{yu2020entangled,aliasgari2020private,d2020gasp,bitar2019private,nodehi2019secure}. Such schemes offer numerical stability but their decoding procedure often relies on inverting a certain matrix which is not scalable with the number of servers. In \cite{jahani2020berrut}, the authors provide a framework for approximately recovering the evaluations of a function, not necessarily a polynomial, over a dataset which is  numerically stable and robust against stragglers. Recently, coded computing schemes have been proposed that enable privacy in the analog domain \cite{soleymani2020privacy, soleymani2020analog}. Also, codes in the analog domain have been recently studied in the context of block codes \cite{roth2020analog} as well as subspace codes \cite{soleymani2019analog} for analog error correction.   

On the other hand, a related work to our RM-coded computing scheme is the recent work in \cite{bartan2019polar} where binary polar codes are applied for distributed matrix multiplication by extending the successive cancellation (SC) decoder of polar codes for real-valued data inputs.
However, to the best of our knowledge, our paper is the first to study coded computation over RM codes. As we will show in this paper (see Section \ref{sec:simulations}), RM-coded computation significantly outperforms polar-coded computation in terms of the average execution time. Despite the observations of excellent performance for RM codes in various disciplines (e.g., capacity-achievability \cite{kudekar2017reed,abbe2015reed} and scaling laws \cite{hassani2018almost}), a critical aspect of RM codes is still the lack of efficient decoding algorithms that are scalable to general code parameters with low complexity. Very recently, \cite{ye2020recursive} proposed a recursive projection-aggregation (RPA) algorithm for decoding RM codes over binary symmetric channels (BSCs) and general binary-input memoryless channels. However, neither \cite{ye2020recursive} nor the earlier works on decoding RM codes \cite{reed1954class,dumer2006soft,saptharishi2017efficiently,santi2018decoding} directly apply to distributed computation over real-valued data. 

\subsection{Our Contributions}\label{sec:contribution}
In this work, we aim at making a strong connection between the problem of characterizing the average execution time of a coded distributed computing system and the fundamental problem of channel coding over erasure channels. The main objective of this paper is twofold: 1) characterizing the fundamental performance limits of coded distributed computing systems employing binary linear codes, and 2) designing practical schemes, building upon binary linear codes, that adapt to the natural constraints imposed by the coded computation applications (e.g., operating over real-valued data) while, in the meantime, achieving very close to the fundamental performance limits with a low complexity.
The main contributions of the paper are summarized as follows.
\begin{itemize}
    \item We connect the problem of characterizing the average execution time of any coded distributed computing system to the error probability of the underlying coding scheme over $n$ uses of erasure channels (see \Lref{lemma1}).
    \item Using the above connection, we characterize the performance limits of distributed computing systems such as the average execution time that any coded computation scheme can achieve (see \Tref{BLC}), the average job completion time using binary random linear codes (see \Cref{RC}), and the best achievable average execution time of a coded computation scheme (see \Cref{col5}) that can, provably, be attained using MDS codes requiring operations over large finite fields.
    \item We establish the existence of binary linear codes that attain, asymptotically, the best performance of a coded computing scheme. This important {result} is established by studying the gap between the average execution time of binary random linear codes and the optimal performance (see \Tref{thm7}), and then showing that the normalized gap approaches zero as $n \rightarrow \infty$ (see \Cref{col8} and \Cref{execution_time_prob}).
    \item By studying the numerical stability of the coded computing schemes utilizing binary linear codes, we show that there exist binary linear codes that are \emph{ numerically stable} against the inevitable rounding errors in practice  while, in the meantime, having an asymptotically optimal average execution time (see \Tref{good_binary_codes_exist}).
    \item We develop an efficient low-complexity algorithm for decoding RM codes over erasure channels. Our decoding algorithm is specifically designed for distributed computation over real-valued data by avoiding any operation over finite fields. Moreover, our decoder is able to achieve very close to the performance of the optimal maximum a posteriori (MAP) decoder (see Section \ref{sec:RM}).
    \item  Enabled by our low-complexity decoder, we study the performance of RM-coded distributed computing systems. We also investigate polar-coded computation.
    \item We carry out extensive numerical analysis confirming our theoretical observations and demonstrating the excellence of RM-coded computation using our proposed decoder.
    
\end{itemize}

The rest of the paper is organized as follows. In Section \ref{sec:sys_model}, we provide the system model and clarify how the system of $n$ independent distributed servers can be viewed as $n$ independent uses of erasure channels. In Section \ref{sec:ave_time}, by connecting the problem of coded computation to the well-established problem of channel coding over erasure channels, we characterize fundamental limits of coded computation using binary linear codes. In Section \ref{sec:stability}, we study the numerical stability of coded computing schemes employing binary linear codes. In Section \ref{sec:practical_codes}, we investigate RM- and polar-coded computation following the presentation of our low-complexity algorithm for decoding RM codes over erasure channels. Finally, we present comprehensive numerical results in Section \ref{sec:simulations}, and conclude the paper in Section \ref{conclusions}.

\section{System Model}\label{sec:sys_model}
We consider a distributed computing system consisting of $n$ local nodes with the same computational capabilities. The run time $T_i$ of each local node $i$ is modeled using a shifted-exponential random variable (RV), mainly adopted in the literature \cite{liang2014tofec,reisizadeh2017coded,lee2018speeding}. Then, when the computational job is equally divided to $k$ tasks, the cumulative distribution function (CDF) of $T_i$ is given by
\begin{align}\label{shifted-exp}
\Pr(T_i\leq t)=1-\exp\left(-\mu(kt-1)\right),~~~\forall ~\!t\geq 1/k,
\end{align}
where $\mu$ is the exponential rate of each local node, also called the straggling parameter. Using \eqref{shifted-exp} one can observe that the probability of the task assigned to the $i$-th server not being completed (equivalent to erasure) until time $t\geq1/k$ is
\begin{align}\label{epst}
\epsilon(t)\triangleq\Pr(T_i>t)=\exp\left(-\mu(kt-1)\right),
\end{align}
and is one for $t<1/k$. Therefore, given any time $t$, the problem of computing $k$ parts of the computational job over $n$ servers can be interpreted as the traditional problem of transmitting $k$ symbols, using an $(n,k)$ code, over $n$ independent-and-identically-distributed (i.i.d.) erasure channels. Note that the form of the CDF in \eqref{shifted-exp} suggests that $t_0\triangleq 1/k$ is the (normalized) deterministic time required for each server to process its assigned $1/k$ portion of the total job (all tasks are erased before $t_0$), while any time elapsed after $t_0$ refers to the stochastic time as a result of servers' statistical behavior (tasks are not completed with probability $\epsilon(t)$ for $t\geq t_0$).

Given a certain code and a corresponding decoder over erasure channels, a \textit{decodable} set of tasks refers to a pattern of unerased symbols resulting in a successful decoding with probability $1$. Then, $P_e(\epsilon,n)$ is defined as the probability of decoding failure over an erasure channel with erasure probability $\epsilon$. For instance, $P_e(\epsilon,1) = \epsilon$ for a $(1,1)$ code. Note that the reason to keep $n$ in the notation is to specify that the number of servers, when the code is used in distributed computation, is also $n$. Finally, the total job completion time $T$ is defined as the time at which a decodable set of tasks/outputs is obtained from the servers.

\section{Fundamental Limits}\label{sec:ave_time}
In this section, we first connect the problem of characterizing the average execution time of any coded distributed computing system to the error probability of the underlying coding scheme over $n$ uses of erasure channels. We then derive several performance limits of coded computing systems such as their average execution time, their average job completion time using binary random linear codes, and their best achievable performance.
Finally, we study the gap between the average execution time of binary random linear codes and the optimal performance to establish the existence of binary linear codes that attain, asymptotically, the best performance of a coded computing scheme.

The following Lemma connects the average execution time of {any} linear-coded distributed computing system to the error probability of the underlying coding scheme over $n$ uses of an erasure channel. 
\begin{lemma}\label{lemma1}
	The average execution time of a coded distributed computing system using a given $(n,k)$ linear code can be characterized as
	\begin{align}
T_{\rm avg}\triangleq\E[T]&=\int_{0}^{\infty}P_e(\epsilon(\tau),n)d\tau\label{ET1}\\
&=\frac{1}{k}+\frac{1}{\mu k}\int_{0}^{1}\frac{P_e(\epsilon,n)}{\epsilon}d\epsilon\label{ET2},
\end{align}
where $\epsilon(\tau)$ is defined in \eqref{epst}.
\end{lemma}
\begin{proof}
It is well-known that the expected value of any {non-negative} RV $T$ is related to its CDF $F_T(\tau)$ as $\E[T]=\int_{0}^{\infty}(1-F_T(\tau))d\tau$. Note that $1-F_T(\tau)=\Pr(T>\tau)$ is the probability of the event that the job is not completed until some time $\tau$. Therefore, using the system model in Section II, we can interpret $\Pr(T>\tau)$ as the probability of decoding failure $P_e(\epsilon(\tau),n)$ of the code when used over $n$ i.i.d. erasure channels with the erasure probability $\epsilon(\tau)$. This completes the proof of \eqref{ET1}. Now given that for the shifted-exponential distribution $d\epsilon(\tau)/d\tau=-\mu k\epsilon(\tau)$, and that $P_e(\epsilon(\tau),n)=1$ for all $\tau\leq1/k$, we have \eq{ET2} by the change of variables.
\end{proof}

\noindent \textbf{Remark 1.} Note that \eqref{ET1} holds given any model for the distribution of the run time of the servers, while \eqref{ET2} is obtained under shifted-exponential distribution, with servers having a same straggling parameter $\mu$, and can be extended to other distributions in a similar approach. 

\begin{theorem}
\label{BLC}
	The average execution time of any coded distributed computing system can be expressed as 
	\begin{align}\label{BL}
	T_{\rm avg}=\frac{1}{k}\left[1+\sum_{i=n-k+1}^{n}\frac{1}{i\mu}\right]+\frac{1}{\mu k}\sum_{i=1}^{n-k}\frac{1}{i}p_{n,k}(i),
	\end{align} 
	where $p_{n,k}(i)$
	is  the average conditional probability  of decoding failure of an $(n,k)$ linear code, for an underlying decoder, given that $i$ encoded symbols are erased at random   where the average is taken over all possible erasure patterns with $i$ erased symbols. \color{black}
\end{theorem}
\begin{proof}
Using the law of total probability and the definition of $p_{n,k}(i)$ we have
\begin{align}\label{BL-MAP}
P_e(\epsilon,n)=\sum_{i=1}^{n}\binom{n}{i}\epsilon^{i}(1-\epsilon)^{n-i}p_{n,k}(i).
\end{align}
Accordingly, characterizing $T_{\rm avg}$ requires computing integrals of the form $f_i\triangleq\int_{0}^{1}\epsilon^{i-1}(1-\epsilon)^{n-i}d\epsilon$ for $i=1,2,...,n$. Using part-by-part integration, one can find the recursive relation $f_{i+1}=\frac{i}{n-i}f_i$ which results in $1/f_i=i\binom{n}{i}$. Note that $p_{n,k}(i)=1$ for $i>n-k$, since one cannot extract the $k$ parts of the original job from less than $k$ encoded symbols. Then plugging \eqref{BL-MAP} into \eqref{ET2} leads to \eqref{BL}.
\end{proof}

Next, we characterize the average execution time using a random ensemble of binary linear codes over $\{\pm 1\}^n$. In this paper, when referred to binary linear codes over real numbers, we always consider codes whose generator matrices only contain $\pm 1$ entries.
The aforementioned random ensemble, denoted by $\cR(n,k)$, is obtained by picking entries of the $k \times n$ generator matrix independently and uniformly at random followed by removing those matrices that do not have a full row rank from the ensemble. 

\noindent\textbf{Remark 2.} Note that \eqref{BL-MAP} together with the integral form in \eqref{ET2} suggest that a coded computing system should always encode with a full-rank generator matrix. Otherwise, the average execution time does not converge. This is the reason behind picking the particular ensemble described above. Note that this is in contrast with the conventional block coding, where we can get an arbitrarily small average probability of error over a random ensemble of all $k \times n$ binary generator matrices.

The following lemma provides an upper bound on the probability that a vector picked from $\{\pm1\}^n$ at random lies in a given subspace of $\R^n$. We utilize this result later to characterize $p_{n,k}(i)$, defined in \Tref{BLC}, for a random code chosen from $\mathcal{R}(n,k)$.

\begin{lemma}\label{prob_bound}
\textup{(\cite[Corollary 4]{kahn1995probability})} For a subspace $V$ of $\R^n$ and $\br$ chosen uniformly at random from $\{\pm1\}^n$ we have 
\be{p_bound}
\Pr \{\br \in V\} \leq 2^{-\dim({V^\perp})},
\ee
where $V^\perp$ is the orthogonal complement of $V$, and $\dim({V^\perp})$ denotes the dimension of ${V^\perp}$.
\end{lemma}

Next, this result is utilized to provide an upper bound on $p_{n,k}(i)$ for a code whose generator matrix is picked from $\mathcal{R}(n,k)$ uniformly at random.

\begin{lemma}
\label{lemma2}
The probability that the generator matrix of a code picked from $\mathcal{R}(n,k)$ does not remain full row rank after erasing $i$ columns uniformly at random, denoted by $p^{\mathcal{R}}_{n,k}(i)$, can be upper bounded as
\begin{align}\label{RCpe}
p^{\mathcal{R}}_{n,k}(i)\leq 1-{\prod_{j=1}^{k}\left(1-2^{j-1-n+i}\right)}.
\end{align}
\end{lemma}
\begin{proof}
Define $l(m,k)$, {$k\leq m$}, as the probability of $k$ signed Bernoulli uniform random vectors $\mathbf{v}_i\in\{\pm1\}^m$ being linearly independent.  Let $V_j$ denote the subspace spanned by $\bv_1, \cdots, \bv_j$. Then one can write 
\begin{align}
    &l(m,j+1)=\\
    &l(m,j)\Pr[\bv_{j+1} \notin V_j| \bv_1, \cdots, \bv_j \ \small{\text{are linearly independent}}] \label{total-prob}\\
    &\geq l(m,j) (1-2^{j-m}),\label{lem}
\end{align}
where \eqref{total-prob} is by the law of total probability and \eqref{lem} is by the result of \Lref{prob_bound}. Note also that $l(m,1)=1> 1-2^{-m}$. Combining this together with \eqref{lem} results in 
\color{black} 
	\begin{align}\label{Pf}
	l(m,k)\geq\prod_{i=1}^{k}\left(1-2^{i-1-m}\right).
	\end{align}
 Note that  	
$
p^{\mathcal{R}}_{n,k}(i)\leq 1-l(n-i,k),
$
since rank-deficient $k \times n$ matrices are already excluded from $\mathcal{R}(n,k)$.
\color{black}
Combining this together with \eqref{Pf} completes the proof.
\end{proof}
\begin{corollary}
\label{RC}
	The average execution time using binary random linear codes from the ensemble $\cR(n,k)$ under maximum a posteriori (MAP) decoding is upper bounded by \eqref{BL} while replacing $p_{n,k}(i)$ in \eqref{BL} by $p^{\mathcal{R}}_{n,k}(i)$, upper bounded in \Lref{lemma2}.
\end{corollary}
\begin{proof}
The proof is by noting that the optimal MAP decoder fails to recover the $k$ input symbols given $n-i$ unerased encoded symbols if and only if the corresponding $k\times(n-i)$ submatrix of the generator matrix of the code is not full row rank which occurs with probability $p^{\mathcal{R}}_{n,k}(i)$.
\end{proof}

\noindent\textbf{Remark 3.} \Tref{BLC} implies that the average execution time using linear codes consists of two terms. The first term is independent of the performance of the underlying coding scheme and is fixed given $k$, $n$, and $\mu$. However, the second term is determined by the error performance of the coding scheme, i.e.,  $p_{n,k}(i)$ for $i=1,2,...,n-k$, and hence, can be minimized by properly designing the coding scheme. 

The following corollary of \Tref{BLC} demonstrates that MDS codes, if they exist,\footnote{ It is in general an open problem whether given $n$, $k$, and $q$, there exists an $(n,k)$ MDS code over $\Fq$ \cite[Ch. 11.2]{macwilliams1977theory}. A non-RS type MDS code construction has been proposed in \cite{roth1989construction}. More recently, a construction of MDS codes with complementary duals has been proposed in \cite{beelen2018explicit} which has received attention due to applications in cryptography \cite{chen2017new,carlet2018euclidean}.  \color{black}} are optimal in the sense that they minimize the average execution time by eliminating the second term of the right hand side in \eqref{BL}. However, for a large number of servers $n$, the field size needs to be also large, e.g., $q > n$ for Reed-Solomon (RS) codes. 

\begin{corollary}[Optimality of MDS Codes]\label{col5}
For given $n$, $k$, and underlying field size $q$, an $(n,k)$ MDS code, if exists, achieves the minimum average execution time that can be attained by any $(n,k)$ code.
\end{corollary}
\begin{proof}
MDS codes have the minimum distance of $d_{\rm min}^{\rm MDS}=n-k+1$ and can recover up to $d_{\rm min}^{\rm MDS}-1=n-k$ erasures leading to $p_{n,k}(i)=0$ for $i=1,2,...,n-k$. Therefore, the second term of \eq{BL} becomes zero for MDS codes and they achieve the following minimum average execution time that can be attained by any $(n,k)$  code:
\begin{align}\label{MDS1}
T_{\rm avg}^{\rm MDS}=\frac{1}{k}+\frac{1}{\mu k}\sum_{i=n-k+1}^{n}\frac{1}{i}.
\end{align} 

\end{proof}

Using \Tref{BLC} and Remark 3, and given that the generator matrix of any $(n,k)$ linear code with minimum distance $d_{\rm min}$ remains full rank after removing up to any $d_{\rm min}-1$ columns, we have the following proposition for the \textit{optimality criterion} in terms of minimizing the average execution time.
\begin{proposition}[Optimality Criterion]\label{prop6}
An $(n,k)$ linear code that minimizes $\sum_{i=d_{\rm min}}^{n-k}{p_{n,k}(i)/i}$ also minimizes the average execution time of a coded distributed computing system.
\end{proposition}
Although MDS codes meet the aforementioned optimality criterion over large field sizes, to the best of our knowledge, the optimal linear codes per \Pref{prop6}, given the field size $q$ and in particular for $q=2$, are not known and have not been studied before, which calls for future studies.

In the following theorem we characterize the gap between the execution time of binary random linear codes and the optimal execution time. Then \Cref{col8} proves that binary random linear codes \textit{asymptotically} achieve the normalized optimal execution time, thereby demonstrating the existence of \textit{good} binary linear codes for distributed computation over real-valued data.
The reason we compare the normalized $nT_{\rm avg}$'s instead of $T_{\rm avg}$'s is that, using \eqref{BL}, $T_{\rm avg}$ has a factor of $1/k$ and hence, $\lim_{n\to\infty}T_{\rm avg}=0$ for a fixed rate\footnote{More precisely, the coding rate over field size $q$ is equal to $k\log_2 q/n$ but with slight abuse of terminology we have dropped the factor of $\log_2 q$ since this factor is not relevant for coded distributed computing.} $R\triangleq k/n >0$.

 \begin{theorem}[Gap of Binary Random Linear Codes to the Optimal Performance]\label{thm7}
Let $T_{\rm avg}^{\rm BRC}$ denote the average execution time of a coded distributed computing system using binary random linear codes. Then, for any given $k$, $n$, we have
\begin{align}
&|nT_{\rm avg}^{\rm MDS}-nT^{\rm BRC}_{\rm avg}|\leq\\&\frac{1}{\mu R}\times\left[\frac{v(n)}{n-k-v(n)+1}+\frac{nR\left(1+\ln\left(n-k-v(n)\right)\right)}{2^{v(n)}}\right],\label{asymp-gap}
\end{align}
where $R$ is the rate and $v(n)$ is an arbitrary function of $n$ with $0 \leq v(n)\leq n-k$.
\end{theorem}
\begin{proof}
Using \Cref{RC} and \Cref{col5}, we have
\begin{align}\label{S}
    \mathcal{S}\triangleq\mu R|nT_{\rm avg}^{\rm MDS}-nT^{\rm BRC}_{\rm avg}|=\sum_{i=1}^{n-k}\frac{1}{i}p^{\mathcal{R}}_{n,k}(i).
\end{align}

To prove the upper bound, the summation in \eqref{S} is split as $\mathcal{S}=\mathcal{S}_1+\mathcal{S}_2$ where
\begin{align}
    \mathcal{S}_1\triangleq\sum_{i=n-k-v(n)+1}^{n-k}\frac{1}{i}p^{\mathcal{R}}_{n,k}(i)&\leq\frac{v(n)}{n-k-v(n)+1},\label{S1,1}
\end{align}
\begin{align}\label{S2}
\mathcal{S}_2\triangleq\sum_{i=1}^{n-k-v(n)}\frac{1}{i}p^{\mathcal{R}}_{n,k}(i).
\end{align}
To upper-bound $\mathcal{S}_2$, we first note that the upper bound on $p^{\mathcal{R}}_{n,k}(i)$, stated in \eqref{RCpe}, is a monotonically increasing function of $i$. Then,
\begin{align}
\mathcal{S}_2&\leq p_f(n-k-v(n),k)\sum_{i=1}^{n-k-v(n)}\frac{1}{i}\label{S2,1}\\
&\leq p_f(n-k-v(n),k)\left(1+\ln\left(n-k-v(n)\right)\right)\label{S2,2},
\end{align}
where \eqref{S2,2} is by the upperbound on the harmonic sum $\sum_{i=1}^{n}\frac{1}{i}\leq 1+ \ln (n)$. \color{black}
We can further upper-bound $p_f(n-k-v(n),k)$ as
\begin{align}
p_f(n-k-v(n),k)&\leq 1-{\prod_{j=1}^{k}(1-2^{j-1-k-v(n)})}\label{pfinf1}\\
&\leq 1-\left[1-2^{-v(n)}\right]^k\label{pfinf2}\\
&\leq nR2^{-v(n)}\label{pfinf3},
\end{align}
where \eqref{pfinf1} is by \eqref{RCpe}, \eqref{pfinf2} follows by noting that
\begin{align}
\prod_{j=1}^{k}(1-2^{j-1-k-v(n)})\hspace{-1mm}=\hspace{-1mm}{\prod_{j'=1}^{k}(1-2^{-j'-v(n)})}\geq[1-2^{-v(n)}]^k,
\end{align}
and \eqref{pfinf3} follows by Bernoulli's inequality $(1-x)^k\geq1-kx$ for any $0<x<1$ and then inserting $k=nR$.
\end{proof}

\begin{corollary}[Asymptotic Optimality of Binary Random Linear Codes]\label{col8}
The normalized average execution time $nT_{\rm avg}^{\rm BRC}$ approaches $nT_{\rm avg}^{\rm MDS}$ as $n$ grows large. More precisely, for a given rate $R$, there exists a constant $c>0$ such that for sufficiently large $n$, i.e., $k=nR$, we have
\begin{align}\label{asymp-gap2}
 nT^{\rm BRC}_{\rm avg}-nT_{\rm avg}^{\rm MDS}\leq c\frac{\log_2 n}{n}.
\end{align}
\end{corollary}
\begin{proof}
Observe that with the choice of $v(n)=2\log_2 n$ both terms in the right hand side of \eq{asymp-gap} become $O(\frac{\log_2 n}{n})$. Note that $n-k=n(1-R) \geq 2\log_2 n$, for sufficiently large $n$. Hence, the upper bound of \eq{asymp-gap2} also holds with a proper choice of $c$.
\end{proof}


\noindent\textbf{Remark 4.}
 \label{optimumrate}
 For any given $n$, one can obtain the optimal value of $k$ and, subsequently, the optimal value of the encoding rate $R$ that minimizes $T_{\rm avg}^{\rm MDS}$ in \eqref{MDS1}. The limit of the optimal value of $R$ when $n \rightarrow \infty$ is referred to as the asymptotically-optimal encoding rate and is denoted by $R^*$.  \color{black}
Using \eqref{MDS1} and a similar approach to \cite{lee2018speeding}, one can show that the asymptotically-optimal encoding rate  $R^*$ for an MDS-coded distributed computing system  is  the solution to 
 \begin{align}\label{RR}
(1-R^*)\ln(1-R^*)=\mu(1-R^*)-R^*.
\end{align}
 \Cref{col8} implies that for distributed computation using binary random  linear codes, the gap between $nT^{\rm BRC}_{\rm avg}$ and $nT^{\rm MDS}_{\rm avg}$ converges to zero as $n$ grows large. Accordingly, the optimal encoding rate also approaches $R^*$, described in \eqref{RR}.

 \begin{corollary}\label{execution_time_prob}
 Let $T^{\rm MDS}$ and $T^{\rm BRC}$ denote the execution time of the coded computing schemes using an MDS code and a code whose generator matrix is picked from $\cR(n,k)$ at random, respectively. Then, there exists a constant $c$ such that 
 \be{prob_bound_time}
 \Pr[n(T^{\rm BRC}-T^{\rm MDS})\geq x]\leq\frac{c\log_2n}{nx}.
 \ee
  \end{corollary}
\begin{proof}
The proof follows immediately by using the result of \Cref{col8} and Markov's inequality. 
\end{proof}



\section{Numerical stability of random binary linear codes \color{black} } \label{sec:stability}
In this section, we study the numerical stability of the coded computing schemes utilizing binary linear codes. Our results indicate  that there exist  binary linear codes that are \emph{ numerically stable} against the inevitable rounding errors in practice which also have asymptotically optimal average execution time.

In general, in a system of linear equations ${\mathbf{A}}{\mathbf{x}}={\mathbf{b}}$, where ${\mathbf{x}}$ is a vector of unknown variables and ${\mathbf{A}}$ is referred to as the \emph{coefficient matrix}, the perturbation in the solution caused by the perturbation in ${\mathbf{b}}$ is characterized as follows. Let $\hat{{\mathbf{b}}}$ denote a noisy version of ${\mathbf{b}}$, where the noise can be caused by round-off errors, truncation, etc. Let also $\hat{{\mathbf{x}}}$ denote the solution to the considered linear system when ${\mathbf{b}}$ is replaced by $\hat{{\mathbf{b}}}$. Let $\Delta {\mathbf{x}}\triangleq \hat{{\mathbf{x}}}-{\mathbf{x}}$ and $\Delta {\mathbf{b}}\triangleq \hat{{\mathbf{b}}}-{\mathbf{b}}$ denote the perturbation in ${\mathbf{x}}$ and ${\mathbf{b}}$, respectively. Then the relative perturbations of ${\mathbf{x}}$ is bounded in terms of that of ${\mathbf{b}}$ as follows \cite{demmel1997applied}:
\be{relative_error}
\frac{\norm{\Delta {\mathbf{x}}}}{\norm{{\mathbf{x}}}}\leq \kappa_{{\mathbf{A}}}\frac{\norm{\Delta {\mathbf{b}}}}{\norm{ {\mathbf{b}}}},
\ee
where $\kappa_{{\mathbf{A}}}$ is the condition number of ${\mathbf{A}}$ and $\norm{\cdot}$ denotes the $l^2$-norm. 

The perturbation bound stated in \eqref{relative_error} implies that the precision loss in the final outcome is $\log_{10} \kappa_{{\mathbf{A}}}$ in decimal digits, where the matrix ${{\mathbf{A}}}$ is the submatirx of the generator whose rows  correspond to non-straggling worker nodes. The precision loss in decoding procedure of the codes constructed over real and complex numbers has been studied in the literature and some codes with deterministic constructions are provided \cite{boley1992algorithmic,ferreira2000stability,ferreira2003stable,henkel1988multiple,marvasti1999efficient}. A code with random Gaussian generator matrix is a numerically stable code with high probability. This is mainly due to the fact that any submatrix of a Gaussian random matrix with i.i.d. entries is also a random Gaussian matrix and such matrices are ill-conditioned only with small probability. The random Gaussian codes are often considered as a benchmark to evaluate the numerical stability of codes over real numbers with explicit constructions \cite{chen2005numerically,chen2005condition,chen2009optimal}. The result of this section implies that the random binary linear codes also offer the same numerical stability as random Gaussian codes. The motivation for using binary linear codes instead of the random Gaussian codes or the existing codes with explicit construction over real numbers is that they can offer a better decoding complexity. This will be clarified further later in Section \ref{sec:simulations} when we compare the performance of practical codes with random codes.          

Let $\cG(n,k)$ denote a random ensemble of  Gaussian codes. This random ensemble is obtained by picking entries of the $k \times n$ generator matrix independently and at random from the standard normal distribution. Let ${\mathbf{G}}$ denote a matrix picked randomly from $\cG(n,k)$ and $\tilde{{\mathbf{G}}}$ denote a random $k \times k$ submatrix of ${\mathbf{G}}$. Note that $\tilde{{\mathbf{G}}}$ is also a random Gaussian matrix. The probability bounds on the condition number of a random Gaussian matrix are provided in \cite{chen2005condition}. In particular, 
\be{Gaussian_condition}
\Pr[\kappa_{\tilde{{\mathbf{G}}}}>nx] < \frac{1}{\sqrt{2\pi}} \frac{C}{x},
\ee
where $5.013<C<6.414$. Consequently, the precision loss in recovery of the computation outcome for the coded computing system using random Gaussian codes is $O(\log_{10} k)$ with high probability. 

The behavior of the largest and smallest singular values of random matrices with i.i.d. entries has been extensively studied in the literature. In particular, we use such results for random matrices with  \emph{sub-Gaussian} random variables. Recall that a random variable $X$ is called sub-Gaussian if its tail is dominated by that of the standard normal random variable, i.e., if there exists $B>0$ such that 

\be{subgaussian_deff}
\Pr[|X|>t] \leq 2 \exp(\frac{-t^2}{B^2}).
\ee
The minimal $B$ is called the \emph{Gaussian moment} of $X$ \cite{rudelson2008littlewood}. Note that a signed Bernoulli random variable $X$ with
$$
\Pr[X=1]=\Pr[X=-1]=\frac{1}{2}
$$
is  sub-Gaussian. Let $\lambda_{\max}$ and $\lambda_{\min}$ respectively denote the largest and the smallest singular value of a random $k\times k$ matrix whose entries are independent zero-mean  sub-Gaussian random variables. Then, 
\be{lambda_max_subgaussian}
\Pr[\lambda_{\max}>Ck^{\frac{1}{2}}+t]\leq 2 \exp{(-ct^2)},
\ee
where $c$ and $C$ are absolute constants \cite{rudelson2010non}.
Moreover, if the variance of the underlying sub-Gaussian random variable is at least $1$, we have
\be{lambda_min_subgaussian}
\Pr[\lambda_{\min}\leq \epsilon k^{-\frac{1}{2}}] \leq C'\epsilon+c'^k
\ee
for all $\epsilon>0$, where $C'$ and $c'$ are constants depending polynomially on the sub-Gaussian moment \cite{rudelson2008littlewood}. It is worth mentioning that the best known $c'$ for the case of Bernoulli random variable is $\frac{1}{\sqrt{2}}+ o(1)$ \cite{bourgain2010singularity}.
 The bounds provided in \eqref{lambda_max_subgaussian} and \eqref{lambda_min_subgaussian} imply that the condition number of a $k \times k$ random Bernoulli matrix is also $O(n)$ with high probability. Hence, the precision loss in recovery of the outcome of a coded computing scheme utilizing random Gaussian codes and random binary linear codes are \emph{almost} the same. This together with the result of \Cref{execution_time_prob} in Section\,\ref{sec:ave_time} imply that there exist binary linear codes that are numerically stable with asymptotically optimal recovery time. This result is stated in the following theorem.
 
 \begin{theorem}\label{good_binary_codes_exist}
Let ${\mathbf{G}}$ denote a $k \times n$  matrix  picked from $\cR(n,k)$ at random. Let also $T^{\rm MDS}$ and $T^{\rm BRC}$ denote the execution time of the coded computing schemes using an MDS code and a code whose generator matrix is ${\mathbf{G}}$, respectively.  Then, the coded computing scheme utilizing the binary linear code generated by ${\mathbf{G}}$ recovers the computation outcome with  $O(\log_{10}\frac{ k}{\epsilon})$ precision loss in decimal digits in $T^{\rm BRC}\leq T^{\rm MDS}+O(\frac{1}{n})$ time with probability $1- O(\epsilon+\frac{\log_2 n}{n})$.  
 \end{theorem}
 
 \begin{proof}
 Let $\tilde{{\mathbf{G}}}$ denote a random $k \times k$ submatix of ${\mathbf{G}}$. Combining \eqref{lambda_max_subgaussian} with \eqref{lambda_min_subgaussian} together with the union bound implies
 \be{union-bound}
 \Pr[\kappa_{\tilde{{\mathbf{G}}}}>\frac{C}{\epsilon}k+\epsilon t\sqrt{k}]\leq 2 \exp{(-ct^2)} +C'\epsilon+c'^k.
 \ee
 Combining \eqref{prob_bound_time}  and \eqref{union-bound} together with the union bound implies
 \begin{align*}
 &\Pr[(\kappa_{\tilde{{\mathbf{G}}}}>\frac{C}{\epsilon}k+\epsilon t\sqrt{k})\ \ \textit{or} \ \ (T^{\rm BRC}-T^{\rm MDS}>\frac{x}{n})]\\
 & \leq \exp{(-ct^2)} +C'\epsilon+c'^k +\frac{c\log_2n}{nx} = O(\epsilon+\frac{\log_2 n}{n})
 \end{align*}
 for all $\epsilon>0$, which completes the proof. 
 \end{proof}
 
 The result of \Tref{good_binary_codes_exist} implies that  there exist numerically stable binary linear codes with asymptotically optimal average execution time. In the rest of the paper, we consider some coding schemes over real numbers constructed based on codes over $\mathbb{F}_2$, namely, RM and polar codes, that offer lower decoding complexity than MDS codes. The numerical stability of such schemes are naturally inherited from the proposed decoding algorithms that involve additions,  subtractions {and, in the case of RM codes, inverting logarithmic-size matrices that are well-conditioned, as numerically verified in the next section}. Moreover, their average execution times are also compared numerically with the optimal values for a wide range of blocklengths (i.e., number of servers).

{
\section{Practical Coded Computing Schemes}\label{sec:practical_codes}
In this section, we explore RM- and polar-coded distributed computation. First, we briefly review RM codes and polar codes, two closely-connected classes of codes, in Section \ref{sec:review_RM}. Then, in Section \ref{sec:RM}, we present our proposed low-complexity algorithm for decoding RM codes over erasure channels that enables RM-coded distributed computing over real-valued data.
Finally, we present polar-coded computation in Section \ref{sec:polar}.

\subsection{Brief Review of RM and Polar Codes}\label{sec:review_RM}
Let $k$ and $n$ be the code dimension and blocklength, respectively, and let $m\triangleq\log_2 n$ be a design parameter. Then, the $r$-th order RM code of length $2^m$, denoted by $\mathcal{RM}(m,r)$, is defined by the following set of vectors as the basis
\begin{align}\label{rm_basis}
\{\mathbf{v}_m(\mathcal{A}):~\mathcal{A}\subseteq[m],|\mathcal{A}|\le r\},
\end{align}
where {$|\mathcal{A}|$ denotes the size of the set $\mathcal{A}$, and $[m]\triangleq\{1,2,\dots,m\}$. Moreover,} $\mathbf{v}_m(\mathcal{A})$ is a row vector of length $2^m$ whose components are indexed by binary vectors $\mathbf{z}=(z_1,z_2,\dots,z_m) \in \{0,1\}^m$. Each component of $\mathbf{v}_m(\mathcal{A})$ is obtained as $\mathbf{v}_m(\mathcal{A},\mathbf{z}) = \prod_{i\in \mathcal{A}} z_i$. In other words, considering a polynomial ring $\mathbb{F}_2[Z_1,Z_2,\dots,Z_m]$ of $m$ variables, the components of $\mathbf{v}_m(\mathcal{A})$ are the evaluations of the monomial $\prod_{i\in \mathcal{A}}Z_i$ at points $\mathbf{z}$ in the vector space $\mathbb{E}\triangleq\mathbb{F}_2^m$. It is easy to observe from \eqref{rm_basis} that there are $\sum_{i=0}^r \binom{m}{i}$ basis (equivalently, $\mathcal{A}$'s) in total, and thus an $\mathcal{RM}(m,r)$ code has a dimension of $k=\sum_{i=0}^r \binom{m}{i}$.

Finally, given the set of basis in \eqref{rm_basis}, the (codebook of) $\mathcal{RM}(m,r)$ code can be defined as the following set of $2^k$ binary vectors
\begin{align}\label{rm_codebook}
\mathcal{RM}(m,r) \triangleq \left\{\sum_{\mathcal{A}\subseteq[m],|\mathcal{A}|\le r}\hspace{-0.5cm}u(\mathcal{A}) \mathbf{v}_m(\mathcal{A}): u(\mathcal{A})\in\{0,1\}\right\}.
\end{align}
Therefore, each codeword $\mathbf{c}=(\mathbf{c}(\mathbf{z}), \mathbf{z}\in\mathbb{E})\in\mathcal{RM}(m,r)$, that is indexed by the binary vectors $\mathbf{z}$, is defined as the evaluations of an $m$-variate polynomial with degree at most $r$ at points $\mathbf{z}\in\mathbb{E}$. 

While RM codes have a universal construction, the construction of polar codes, on the other hand, is \textit{channel-specific}. Consider Ar{\i}kan's $n\times n$ polarization matrix $\mathbf{G}_n=\begin{bmatrix}
1 & 0\\ 
1 & 1
\end{bmatrix}^{\otimes m}$, where $m=\log_2 n$ and $\mathbf{A}^{\otimes m}$ denotes the $m$-th Kronecker power of $\mathbf{A}$. The encoding of polar codes is obtained from the aforementioned polarization matrix $\mathbf{G}_n$ in a {channel-specific} manner. Particularly, in the case of binary erasure channels (BECs), a design parameter $\epsilon_d$ is picked, as specified later in Section \ref{sec:simulations}. Then the polarization transform $\mathbf{G}_n$ is applied to a BEC with erasure probability $\epsilon_d$, BEC$(\epsilon_d)$. The erasure probabilities of the polarized bit-channels, denoted by $\{Z_i\}_{i=1}^n$, are sorted and the $k$ rows of $\mathbf{G}_n$ corresponding to the indices of the $k$ smallest $Z_i$'s are picked to construct the $k\times n$ generator matrix $\mathbf{G}$.

One can also obtain an equivalent encoding of $\mathcal{RM}(m,r)$, similar to that of polar codes, by selecting rows of the square matrix $\mathbf{G}_n$ that have a Hamming weight of at least $2^{m-r}$. In this case, the resulting generator matrix $\boldsymbol{G}_{k\times n}$ will have $\binom{m}{i}$ rows of Hamming weight $n/2^i$, for $i=0,1,\cdots, r$.

\subsection{RM-Coded Distributed Computation}\label{sec:RM}
It has recently been shown that RM codes are capacity achieving over BECs \cite{kudekar2017reed}, though under bit-MAP decoding, and numerical results suggest that they actually achieve the capacity with \textit{almost} optimal scaling \cite{hassani2018almost}.
RM codes also achieve the capacity of BSCs at extreme rates, i.e., at rates converging to zero or one \cite{abbe2015reed}. They are also conjectured to have characteristics similar to those of random codes in terms of both scaling laws \cite{hassani2018almost} and weight enumeration \cite{kaufman2012weight}. Despite all these excellent properties, RM codes still lack efficient low-complexity decoders for general code dimensions and blocklengths. Very recently, Ye and Abbe \cite{ye2020recursive} proposed a recursive projection-aggregation (RPA) algorithm for decoding RM codes over BSCs and general binary-input memoryless channels. The RPA algorithm is comprised of three main steps: 1) projecting the received corrupted codeword onto the cosets defined by each projection subspace, 2) recursively decoding the projected codewords, and 3) aggregating the decoded codewords at the next layer with the current observation to finally decode the original RM codeword.

In this section, we propose an efficient RPA-like algorithm for decoding RM codes over erasure channels. Our decoding algorithm has three major novelties. First, it only involves additions, subtractions, {and inverting relatively small matrices of size no more than $\log n+1$.}
We need to emphasize that the RPA algorithms proposed in \cite{ye2020recursive} work over the binary field and do not directly apply to real-valued inputs. For example, as detailed in Section \ref{sec:RM_projection}, the projection step in the original RPA algorithms requires addition of the received bits over the binary field (i.e., XOR'ing them). Therefore, our proposed decoding algorithm generalizes the RPA algorithms to the case of erasure channels while avoiding operations over finite fields.
Second, our decoding algorithm has a low complexity achieved by carefully selecting a small fraction of the total number of projections, i.e., only $\binom{m}{r-1}$ projection subspaces of dimension $s=r-1$ are selected to decode an $\mathcal{RM}(m,r)$ code. Therefore, our decoding algorithm enables decoding RM codes of higher orders and lengths with a manageable complexity. Third, our simulation results suggest that our decoding algorithm is able to achieve very close to the performance of optimal MAP decoding while maintaining a low complexity.

In the following, we explain our decoding algorithm that is comprised of three main steps separately described in Sections \ref{sec:RM_projection}, \ref{sec:RM_decoding}, and \ref{sec:RM_aggregation}. 
Assuming an $\mathcal{RM}(m,r)$ code, our decoding algorithm only applies one layer of projection
using $s$-dimensional subspaces with $s=r-1$. It then applies MAP decoding (block-MAP in conjunction with bit-MAP) to decode the projected codewords,
and finally aggregates the results to recover the codeword. 
{As we show in Section \ref{sec:RM_projection}, a smaller value of $s$ increases the complexity of the decodings performed at the bottom layer but it can also result in a better decoding performance since the MAP decoding over projected vectors will be utilized to a greater extent (recall that there is only a single layer of projection). It will be clarified later that this choice of $s$ can result in a manageable decoding complexity while not much sacrificing the decoding performance.}

\subsubsection{Projection}\label{sec:RM_projection} The decoding algorithm starts by projecting the received corrupted codeword onto the cosets defined by the projection subspaces. More specifically, let $\mathbb{B}_i$ be a $s$-dimensional subspace of $\mathbb{E}$, with $s\leq r$ and $i\in[\mathcal{I}]$ where $\mathcal{I}$ is the total number of projection subspaces. The quotient space $\mathbb{E}/\mathbb{B}_i$ contains all the cosets of $\mathbb{B}_i$ in $\mathbb{E}$ where each coset $\boldsymbol{\tau}$ has the form $\boldsymbol{\tau}=\mathbf{z}+\mathbb{B}$ for some $\mathbf{z}\in\mathbb{E}$. Given a length-$n$ codeword and a $s$-dimensional subspace, the objective of the projection step is to obtain a length-$n/2^s$ vector whose each component, that corresponds to one of the $n/2^s$ cosets, is an \textit{appropriate} representative of the bits indexed by the $2^s$ elements of that coset. Building upon this notion and the definition of the log-likelihood ratio (LLR), the expressions for the projection of the channel LLRs over general binary-input memoryless channels are derived in \cite{ye2020recursive} for subspaces of dimension $s=1$ and $2$. By applying the same principles, one can show that the projection of the length-$n$ vector of the corrupted codeword over a BEC can be obtained by adding the bits indexed by the cosets over the binary field (i.e., XOR'ing) while assuming the addition of bits $0$ or $1$ with an erasure is also an erasure. In other words, considering $\mathbf{y}=(\mathbf{y}(\mathbf{z}), \mathbf{z}\in\mathbb{E})$ as the length-$n$ vector of the corrupted codeword over a BEC, the projected vector onto a subspace $\mathbb{B}$ can be obtained as $\mathbf{y}_{/\mathbb{B}} \triangleq\big(\mathbf{ y}_{/\mathbb{B}}(\boldsymbol{\tau}), \boldsymbol{\tau}\in\mathbb{E}/\mathbb{B} \big)$ such that $\mathbf{ y}_{/\mathbb{B}}(\boldsymbol{\tau}) \triangleq \bigoplus_{\mathbf{z}\in \boldsymbol{\tau}} \mathbf{ y}(\mathbf{z})$, where $\bigoplus$ denotes the coordinate-wise addition in $\mathbb{F}_2$ while assuming that additions with an erasure is also an erasure.

While the projection rule  $\mathbf{ y}_{/\mathbb{B}}(\boldsymbol{\tau})= \bigoplus_{\mathbf{z}\in \boldsymbol{\tau}} \mathbf{ y}(\mathbf{z})$  works perfectly over the binary field, we are looking for a low-complexity decoder that works over real numbers to avoid numerical issues (caused by working over finite fields). To this end, considering the received vector as $\mathbf{y}=(\mathbf{y}(\mathbf{z}), \mathbf{z}\in\mathbb{E})$, we can obtain the projected vector by linearly combining (over real numbers) the entries indexed by the cosets to obtain $\mathbf{y}_{/\mathbb{B}} \triangleq\big(\mathbf{ y}_{/\mathbb{B}}(\boldsymbol{\tau}), \boldsymbol{\tau}\in\mathbb{E}/\mathbb{B} \big)$ such that
\begin{align}\label{eq_proj}
\mathbf{ y}_{/\mathbb{B}}(\boldsymbol{\tau}) \triangleq \sum_{\{i:\mathbf{z}_i\in \boldsymbol{\tau}\}} \gamma_i\mathbf{ y}(\mathbf{z}_i),
\end{align}
where $\gamma_i$'s are some properly-chosen real-valued combining coefficients, as clarified later. 
As clarified in \cite[Remark 1]{jamali2021reed}, the result of the projection operation can be thought of as obtaining a generator matrix that is formed by merging the columns of the original code generator matrix indexed by the cosets of the projection subspace. In other words, given a $k\times n$ generator matrix $\mathbf{G}_{k\times n}$, we can define matrices of dimension $k\times n/2^s$, referred to as \textit{projected generator matrices}, each obtained by merging the columns of the original generator matrix indexed by the cosets of each projection subspace. In this paper, assuming an $\mathcal{RM}(m,r)$ code, we work based on projection subspaces of dimension $s=r-1$. In the binary field, this choice of $s$ will result in order-$1$ RM codes, i.e., $\mathcal{RM}(m-r+1,1)$ codes, at the bottom layer after the projection \cite{ye2020recursive}. These codes all have the same dimension of $m-r+2$ that is also equal to the rank of the projected generator matrices. However, for our decoder that works on real-valued inputs the aforementioned projected generator matrices can have different ranks and they no longer (necessarily) correspond to the generator matrices of the  lower order (and lower length) RM codes. 

{ Next, we discuss how to choose $\gamma_i$'s in \eqref{eq_proj}. In our proposed decoder, we pick $\gamma_i \in \{-1,+1\}$ such that combining coefficients result in the same rank of $m-r+2$ as the lower order RM code $\mathcal{RM}(m-r+1,1)$ after an $(r-1)$-dimensional projection of $\mathcal{RM}(m,r)$. To illustrate this, and for simplicity, we consider the one-dimensional projections on $\mathcal{RM}(m,r)$. Note that the argument for general $s$-dimensional projections, and, consequently, for $s=r-1$, would follow naturally by cascading $s$ one-dimensional projections. Note that there are, at least, $m$ projections of $\mathcal{RM}(m,r)$ over real numbers that lead to the $\mathcal{RM}(m-1,r-1)$ code. These projections correspond to
standard basis vectors $\boldsymbol{e}_j$'s, where $\boldsymbol{e}_j$, for $j\in[m]$, is a length-$m$ vector with a $1$ in the $j$-th position and $0$ in all other $m-1$ positions. For instance, the projection corresponding to $\boldsymbol{e}_m = (0,0,\dots,1)$ leads to pairing of consecutive column indices of the generator matrix of $\mathcal{RM}(m,r)$, i.e., $(0,1), (2,3), ..., (n-2,n-1)$. The main observation, which enables this whole process, is that for each pair $(2i,2i+1)$, the support of the column $2i+1$ is a subset of the support of the column $2i$ (when considering entries to be $0$'s and $1$'s). Therefore, if we subtract column $2i+1$ from column $2i$, in the real domain, that would correspond to XORing them in the binary domain. Hence, we get the generator matrix of $\mathcal{RM}(m-1,r-1)$ with this particular choice of $\gamma_i$'s for each of these $m$ projections. Also, a simple normalization by $2$ is done in order to keep matrix entries belonging to the set $\{-1,+1\}$. For general $s$, one can choose $\mathcal{I}=\binom{m}{s}$ $s$-dimensional subspaces $\mathbb{B}_i$ each obtained by the span of $s$ standard basis vectors $\boldsymbol{e}_j$'s and apply the aforementioned process recursively. In particular, for $s = r-1$, we can efficiently find the choices of $\gamma_i$'s in $\{-1,+1\}$ that lead to projection of $\mathcal{RM}(m,r)$ into $\mathcal{RM}(m-r+1,1)$ using ${m \choose r-1}$ projections.}

Note that, given an $\mathcal{RM}(m,r)$ code for the coded distributed computation, the set of the projection subspaces, the cosets, the combining coefficients (that result in the rank of $m-r+2$), and the corresponding projected generator matrices will be computed before hand to lower the complexity of the decoder and prevent any operation over finite fields during the decoding process. It is worth mentioning at the end that the specific selection of the combining coefficients as described above has several advantages. First, having the combining coefficients being either $-1$ or $+1$ renders a projection step that only involves additions and subtractions (and no multiplications). It also equally weights all the entries to be combined which in turn can prevent numerical stability issues that may arise from some of the weights being very large or very small. Second, the corresponding generator matrices all will have the smallest possible rank. This is a twofold gain: 1) it lowers the complexity of the MAP decoding over the projected vectors, and 2) it can also improve the error rate performance of the decodings at the bottom layer by increasing the chances of getting the same rank for the projected generator matrices after erasing some of the columns (see Section \ref{sec:RM_decoding}).

\subsubsection{Decoding of the Projected Vectors}\label{sec:RM_decoding}
Once the projection step is completed, we end up with $\mathcal{I}$ length-$2^{m-r+1}$ projected vectors $\mathbf{y}_{/\mathbb{B}_i}$'s, $i\in[\mathcal{I}]$, each corresponding to a $k\times 2^{m-r+1}$  projected generator matrix of rank $m-r+2$. Given that we are working over erasure channels, we either know an entry perfectly (without error) or we do not know it at all. The objective of this step is to decode the projected vectors by decoding either all the erased entries (block-MAP) or a fraction of them (bit-MAP) in each of these projected vectors. For each $\mathbf{y}_{/\mathbb{B}_i}$, let $\mathcal{Y}_{i}^e$ and $\mathcal{Y}_i^{ne}$ be the sets of the indices of the erased and non-erased entries, respectively. Also let $\mathbf{G}^{ne}_{p,i}$ denote the $k\times|\mathcal{Y}_i^{ne}|$ submatrix of the projected generator matrix (over the subspace $\mathbb{B}_i$) comprised of the $|\mathcal{Y}_i^{ne}|$ columns corresponding to non-erased entries in $\mathbf{y}_{/\mathbb{B}_i}$. Then if the rank of $\mathbf{G}^{ne}_{p,i}$ is equal to $m-r+2$ (i.e., if erasing $|\mathcal{Y}_i^{e}|$ columns of the projected generator matrix does not change its rank), the block-MAP decoder can be applied to decode all the erased entries. This is equivalent to say that each of the erased entries can be obtained as a linear combination of the non-erased entries.
On the other hand, if the block-MAP decoder fails, there is still a chance that the bit-MAP decoder can decode some of the entries. In particular, 
{if rank remains the same after adding to $\mathbf{G}_{p,i}^{ne}$, a column
corresponding to one of the $|\mathcal{Y}_i^{e}|$ indices, then that particular erased entry can be recovered as a linear combination of non-erased entries.}
By performing the same procedure for all $\mathbf{y}_{/\mathbb{B}_i}$'s, we obtain the decoded projected vectors $\mathbf{\hat{y}}_{/\mathbb{B}_i}$'s for all $i\in[\mathcal{I}]$.

\subsubsection{Aggregation}\label{sec:RM_aggregation} The objective in this step is to combine the observation from the channel output, i.e., the corrupted vector $\mathbf{y}$, with that of the projected vectors, i.e., $\mathbf{\hat{y}}_{/\mathbb{B}_i}$'s, to obtain the decoded vector $\mathbf{\hat{y}}$. Let $\mathcal{Y}^e$ be the set of the indices of the erased entries in $\mathbf{y}$. In the following, we describe our aggregation method assuming a given projection subspace $\mathbb{B}_i$, $i\in[\mathcal{I}]$, and an erased index $l\in\mathcal{Y}^e$. By applying the same procedure for all projection subspaces and all erased indices, we can obtain the decoded vector $\mathbf{\hat{y}}$.

Let $\mathbf{z}$ be the binary vector indexing the $l$-th position which is an erasure. Also, let $[\mathbf{z}+\mathbb{B}_i]$ denote the cost of $\mathbb{B}_i$ that contains $\mathbf{z}$.
Since $\mathbb{B}_i$ is an $s$-dimensional subspace (where $s=r-1$ in this paper), there are $2^s-1$ elements, denoted by $\mathbf{z}_1,\mathbf{z}_2, \cdots,\mathbf{z}_{2^s-1}$, in $[\mathbf{z}+\mathbb{B}_i]$ apart from $\mathbf{z}$ itself. Given the decoded projected vector $\mathbf{\hat{y}}_{/\mathbb{B}_i}$, we know from the projection step that $\mathbf{\hat{y}}_{/\mathbb{B}_i}([\mathbf{z}+\mathbb{B}_i])$ is an estimation of $\gamma \mathbf{y}(\mathbf{z})+\sum_{j=1}^{2^s-1}\gamma_j \mathbf{y}(\mathbf{z}_j)$, where $\{\gamma,\gamma_1,\cdots,\gamma_{2^s-1}\}$ is the set of the predefined combining coefficients for this particular projection. Now, if $\mathbf{\hat{y}}_{/\mathbb{B}_i}([\mathbf{z}+\mathbb{B}_i])$ and all $\mathbf{y}(\mathbf{z}_j)$'s, for $j=1,\cdots 2^s-1$, are known (i.e., none of them are erased), we can recover $\mathbf{y}(\mathbf{z})$ as
\begin{align}\label{eq_aggr}
    \mathbf{\hat{y}}(\mathbf{z})=\frac{1}{\gamma}\bigg[\mathbf{\hat{y}}_{/\mathbb{B}_i}([\mathbf{z}+\mathbb{B}_i])-\sum_{j=1}^{2^s-1}\gamma_j \mathbf{y}(\mathbf{z}_j) \bigg].
\end{align}
Note that since we are working over erasure channels, if any  projection satisfies the above condition, we do not need to check the other projections for decoding the erased entry at the $l$-th index.  

{
\noindent\textbf{Example 1.}
Consider a corrupted codeword of $\mathcal{RM}(3,2)$, denoted by $\mathbf{y}=(y_1,y_2,\cdots,y_8)$, with an erasure at the fifth position. Also, we consider three projection subspaces as $\mathbb{B}_1=\{(0,0,0),(1,0,0)\}$, $\mathbb{B}_2=\{(0,0,0),(0,1,0)\}$, and $\mathbb{B}_3=\{(0,0,0),(0,0,1)\}$. Consider projection in the direction of $\mathbb{B}_1$. For this projection, the quotient space is $\mathbb{E}/\mathbb{B}_1=\{(1,5),(2,6),(3,7),(4,8)\}$. Additionally, one can observe that the choice of combining coefficients $\gamma_1=-1$ and $\gamma_2=1$ results in the minimum rank of $m-r+2=3$ for all three projected generator matrices. Accordingly, the projected received vector $\mathbf{y}/\mathbb{B}_1$ is of the form $(y_5-y_1,y_6-y_2,y_7-y_3,y_8-y_4)$ which will have an erasure only at the first position. Since the corresponding projected generator matrix remains rank-$3$ after removing the first column, the block-MAP decoder will be able to recover the single erasure in $\mathbf{y}/\mathbb{B}_1$. Similarly, one can show that the other two projected received vectors $\mathbf{y}/\mathbb{B}_2$ and $\mathbf{y}/\mathbb{B}_3$ will both have a single erasure at the third position. Given that their corresponding projected generator matrices also remains rank-$3$ after removing the third column, the decoder will be able to recover the erased bits at the third indices. Finally, given the successful decoding of the projected vectors, one can successfully decode the corrupted codeword after aggregation.
}

Figure \ref{F1_BLER} compares the block error rate (BLER) results for our projective decoder with that of optimal MAP decoder for various RM codes of interest to distributed computing (see Table \ref{table} for the rationale behind the code parameters in this figure). It is remarkable that our decoder can achieve almost the same performance as that of the MAP decoder with only $\binom{m}{r-1}$ projections for an $\mathcal{RM}(m,r)$ code\footnote{{We have observed that the gap between the performance of our low-complexity decoding algorithm and that of MAP increases as we increase $n$ (e.g., in $\mathcal{RM}(7,4)$ and $\mathcal{RM}(8,5)$ codes that have close code parameters to the RM subcodes considered in Table I). However, very large $n$'s can be less relevant in the context of distributed computing as $n$ here corresponds to the number of servers.}}. In fact, only $\mathcal{I}=3$, $4$, $10$, and $15$ projection subspaces are selected for $\mathcal{RM}(3,2)$, $\mathcal{RM}(4,2)$, $\mathcal{RM}(5,3)$, and $\mathcal{RM}(6,3)$, respectively, which are significantly less than the full number of projections for RPA-like decoding of these codes. Note that the decoder may also iterate the whole process, described in this section, a few times to ensure the convergence of the algorithm. {The convergence here means that there is no difference between what is known about the codeword, in terms of the corrected symbols, at the end of the current iteration with that of the previous iteration.} Note that the number of outer iterations, denoted by $N_{\rm max}$, is not much to cause a serious complexity issue for our algorithm. In particular, in Figure \ref{F1_BLER}, the maximum number of outer iterations is chosen to be $N_{\rm max}=1$, $2$, $2$, and $3$ for  $\mathcal{RM}(3,2)$, $\mathcal{RM}(4,2)$, $\mathcal{RM}(5,3)$, and $\mathcal{RM}(6,3)$, respectively.

\begin{figure}
	\centering
	\includegraphics[trim={0cm 0.2cm 0cm 0cm},width=3.6in]{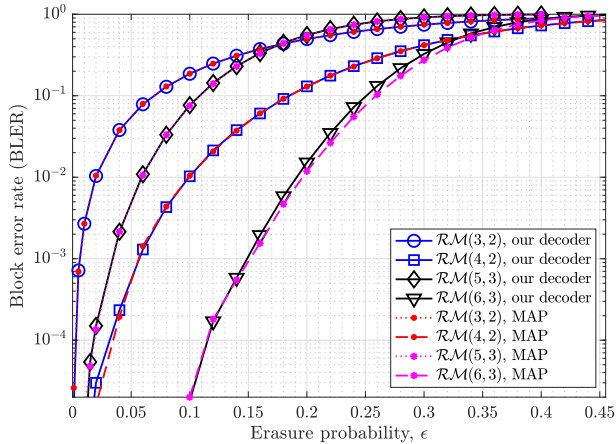}
	\caption{Block error rate (BLER) results for the MAP decoder and our projective decoder for various RM codes of interest to distributed computing. The maximum number of outer iterations is chosen to be $N_{\rm max}=1$, $2$, $2$, and $3$ for $\mathcal{RM}(3,2)$, $\mathcal{RM}(4,2)$, $\mathcal{RM}(5,3)$, and $\mathcal{RM}(6,3)$, respectively.}
	\label{F1_BLER}
\end{figure}

{Next, we discuss the numerical stability of the proposed RM-coded computing system. As discussed earlier, the decoding algorithm only involves additions and subtractions, as well as inverting small matrices of size at most $\log n +1$. To make sure that these matrices are well-conditioned, we did extensive numerical analysis by looking at the condition numbers of the matrices involved in the inversion. As explained in the paper, if a given projected generator matrix has rank $m-r+2$ after removing the columns specified by the projected corrupted codeword, then the block-MAP decoder will be able to recover the erased indices in the projected vectors. To do so, there are several methods and, in our decoder, we compute the inverse of the matrix $\tilde{\mathbf{G}}\tilde{\mathbf{G}}^{\text{T}}$, where $.^{\text{T}}$ denotes the matrix transpose operation and  $\tilde{\mathbf{G}}$ is a full-rank sub-matrix of the projected generator matrix (excluding the erased columns) obtained by selecting a subset of $m-r+2$ linearly independent rows. Our numerical analysis indicates that the matrices $\tilde{\mathbf{G}}\tilde{\mathbf{G}}^{\text{T}}$ involved in the inversion have relatively small condition numbers and thus are well-conditioned.
For instance, we considered decoding $\mathcal{RM}(6,3)$ using $\binom{6}{3-1}=15$ two-dimensional projections as explained in the paper. We also considered $60$ equidistant erasure probabilities in the interval $[0.01,0.6]$, and examined $1000$ random erasure patterns for each of the erasure probabilities. We then projected each erasure pattern (i.e., corrupted codeword) onto the $15$ projections, and computed the condition number of full-rank $5\times 5$ matrices $\tilde{\mathbf{G}}\tilde{\mathbf{G}}^{\text{T}}$ after each projection. In our numerical results the maximum, i.e., the worst case, condition number was $428.36$ among all these $1000\times 60\times 15=9\times10^5$, which is then equivalent to loosing no more than $3$ precision digits in decimal floating-point representation. Also, the average of the condition numbers among all $1000$ random trails for each given erasure probability and projection subspaces was smaller than $100$.}



\subsection{Polar-Coded Distributed Computation}\label{sec:polar}
Binary polar codes are capacity-achieving linear codes with explicit constructions and low-complexity encoding and decoding \cite{Arikan}. Also, the low-complexity $O(n\log n)$ encoding and decoding of polar codes can be adapted to work over real-valued data when dealing with erasures as in coded computation systems. Given the close connection of RM and polar codes, for the sake of the completeness of our study, we also explore polar-coded computation{, which was first considered in \cite{bartan2019polar}}. However, our simulation results in Section \ref{sec:simulations} demonstrate a significantly superior performance for the RM-coded computation enabled by our proposed low-complexity decoder in Section \ref{sec:RM}.
Next, we briefly explain the encoding and decoding procedure of real-valued data using binary polar codes.


\subsubsection{Encoding Procedure}\label{sec:polar_encoding} 
In Section \ref{sec:review_RM}, we briefly explained the construction of the generator matrix for polar codes over BECs.
The encoding procedure using the resulting $k\times n$ generator matrix $\mathbf{G}$, which also applies to any $(n,k)$ binary linear code operating over real-valued data, is as follows. First, the computational job is divided into $k$ smaller tasks. Then the $j$-th encoded task which will be sent to the $j$-th node, for $j=1,2,\dots,n$, is 
the linear combination of all tasks according to the $j$-th column of $\mathbf{G}$.
Throughout the paper (including both RM- and polar-coded computation), we apply the transform $\mathbf{G}\rightarrow 2\mathbf{G}-1$ to convert the entries of the generator matrix from $\{0,1\}$ to $\{-1,+1\}$.
\subsubsection{Decoding Procedure}\label{sec:polar_decoding}
The recursive structure of polar codes can be applied for low-complexity detection/decoding of real-valued data using parallel processing for more speedups \cite{jamali2018low,jamali2020massive}. {To this end, one can apply the decoding algorithms in \cite{bartan2019polar} for polar-coded computing.} It is well-known that in the case of SC decoding over erasure channels, the probability of decoding failure of polar codes is $P_e^{\rm SC}(\epsilon,n)=1-\prod_{i\in\mathcal{B}}(1-Z_i)$, where $\mathcal{B}$ denotes the set of indices of the selected rows.

\section{Simulation Results}\label{sec:simulations}
In this section, simulation results for the execution time of various coded distributed computing schemes are presented. In particular, their gap to the optimal performance are shown and also, their performance gains are compared with the uncoded computation. In addition to the fundamental results presented in Section III, we also provide the simulation results for RM- and polar-coded distributed computation, explained in the previous section. We assume $\mu=1$ for all numerical results in this section. 

\begin{table*}[t]
     		\centering
     		\caption{Average execution time and optimal $k^*$ values for different coding schemes as well as their gap $g_{\rm opt}$ to the optimal performance and their performance improvement gain ${G}_{\rm cod}$ compared to the uncoded computing.}
     		\label{table}\vspace{-0.2cm}
   		\begin{tabular}{||M{0.08in}||M{0.50in}|M{0.7in}|M{1in}|M{1in}|M{0.95in}|M{1in}||}
     			\hline\hline
     \vspace{0.1cm}	$\!\!n$ & 
    {Uncoded}&
     {MDS coding}&
    {Binary random coding}&
     {Polar coding with SC decoding}&
    {RM coding with our proposed decoder}&
    {RM coding with optimal MAP decoding}\\
     			\cline{2-7}
    & $(T_{\rm avg},g_{\rm opt})$ &
     $(T_{\rm avg},k^*,{G}_{\rm cod})$ & 
     $(T_{\rm avg},k^*,g_{\rm opt},{G}_{\rm cod})$ & 
    $(T_{\rm avg},k^*,g_{\rm opt},{G}_{\rm cod})$ &
    $\!(T_{\rm avg},k^*,g_{\rm opt},{G}_{\rm cod})$
    &
    $(T_{\rm avg},k^*,g_{\rm opt},{G}_{\rm cod})$
    \\ \hline	\hline	
     			$\!\!8$  & $\!\!(0.4647,25\%)$& $(0.370,6,20\%)$&$(0.460,7,25\%,1.1\%)$&$(0.412,7,11\%,12\%)$&$(0.389,7,5.1\%,16\%)$&$(0.389,7,5.1\%,16\%)$\\\hline
     			$\!\!16$  & $\!\!(0.2738,44\%)$& $(0.191,11,31\%)$&$(0.226,11,18\%,18\%)$&$(0.217,11,14\%,21\%)$&$\!(0.198,11,3.6\%,28\%)$&$\!(0.198,11,3.6\%,28\%)$\\ \hline
     				$\!\!32$  & $\!\!(0.1581,63\%)$& $(0.0968,22,39\%)$&$(0.105,21,8.6\%,34\%)$&$(0.114,24,18\%,28\%)$&$\!(0.104,26,7.2\%,34\%)$&$\!(0.104,26,7.2\%,34\%)$\\ \hline
     					$\!\!64$  & $\!\!(0.0897,84\%)$& $(0.0488,44,46\%)$&$(0.051,43,3.9\%,44\%)$&$(0.0584,44,20\%,35\%)$&$\!\!\!(0.0506,42,3.7\%,44\%)$&$\!(0.050,42,2.6\%,44\%)$\\ \hline
     						$\!\!\!128$  & $\!\!\!(0.0503,\!105\%)$& $\!(0.0245,88,51\%)$&$\!(0.025,87,1.9\%,50\%)$&$\!(0.0293,88,19\%,42\%)$&---&$\!(0.0252,97,2.8\%,50\%)$\\ \hline
     							$\!\!\!256$  & $\!\!\!(0.0278,\!127\%)$& $\!\!(0.0123,\!175,\!56\%)$&$\!\!(0.0124,174,0.9\%,56\%)$&$\!\!(0.0146,182,19\%,48\%)$&---&$\!\!\!(0.0123,166,0.6\%,56\%)$\\ \hline
     								$\!\!\!512$  & $\!\!\!(0.0153,\!149\%)$& $\!\!\!(0.0061,\!350,\!60\%)$&$\!\!\!(0.0062,349,0.5\%,60\%)$&$\!\!(0.0073,388,19\%,52\%)$&---&$\!\!\!(0.0061,353,0.1\%,60\%)$\\ \hline\hline
\end{tabular}
\end{table*}

For MDS and random linear codes, $T_{\rm avg}$ is calculated  using  \eqref{MDS1} and \Cref{RC}, respectively.
 For the polar-coded computation with SC decoding, we apply $P_e^{\rm SC}(\epsilon,n)$, explained in Section \ref{sec:polar_decoding}, together with \eqref{ET2} to evaluate $T_{\rm avg}$ via the numerical integration. Similarly, for the RM-coded computation with our decoder, we first apply our decoding algorithm, presented in Section \ref{sec:RM} to numerically obtain $P_e(\epsilon,n)$ given $\epsilon$ and $n$. We then apply \eqref{ET2} to numerically calculate $T_{\rm avg}$. Note that, for each scheme, we searched over all possible values of $k$ to obtain the optimal $k^*$ that minimizes $T_{\rm avg}$. For the RM-coded computation, we also include the results for the optimal MAP decoder. To do so, given $k$ and $n$, we first construct the generator matrix of the RM (sub-) code by selecting the $k$ rows of the polarization matrix $\mathbf{G}_n$ that have the highest Hamming weights. We then apply Monte-Carlo numerical simulation to obtain $P_e(\epsilon,n)$ (we erase the columns of the resulting generator matrix with probability $\epsilon$ and then declare an error if the resulting matrix is not full rank) for each value of $\epsilon$ considered for the numerical evaluation of $T_{\rm avg}$ via \eqref{ET2}. It is worth mentioning that we used the same $k^*$ values for the RM-coded computation with our decoding algorithm as that of the MAP decoding. In fact, we only evaluate the performance of our decoder if, given $n=2^m$, the value of $k^*$ is of the form $\sum_{i=0}^r\binom{m}{i}$ for some $r\leq m$. This is because our decoding algorithm is specifically designed for RM codes and not their subcodes. Therefore, one needs to build upon the methods in this paper and the intuitions in \cite{jamali2021reed} to extend our decoding algorithm for general RM subcodes (with a carefully designed generator matrix) that admit any code dimension $k$. Given the close performance of our decoder for RM codes to that of the optimal MAP decoder (and also similarities in decoding RM codes and their subcodes), we expect the generalized version of our decoder to the case of RM subcodes to also achieve very close to the optimal performance. 

  \begin{figure}
	\centering
	\includegraphics[trim={0cm 0.2cm 0cm 0cm},width=3.6in]{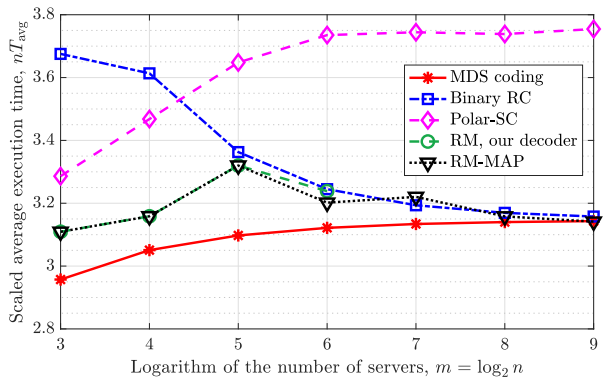}
	\caption{Scaled average execution time of a homogeneous distributed computing system with $\mu=1$ using various coding schemes for finite number of servers $n=8$, $16$, $32$, $64$, $128$, $256$, and $512$.}
	\label{F1}
\end{figure}

 Numerical results for the performance  of the coded distributed computing systems utilizing MDS codes, binary  random linear codes, polar codes, and RM codes, are presented in Table \ref{table} and are compared with the uncoded scenario over  small-to-moderate blocklengths. 
 We designed the polar code with  $\epsilon_d =0.1$, which is observed to be good enough for the range of blocklengths in Table \ref{table}. One can also attain slightly better performance for  polar codes by optimizing over $\epsilon_d$ specifically for each $n$. Characterizing the best  $\epsilon_d$ as a function of  blocklength $n$ is left for the future work. In Table I, $G_{\rm cod}$ is defined as the percentage of the gain in $T_{\rm avg}$ compared to the uncoded scenario and $g_{\rm opt}$ is defined as the gap of $T_{\rm avg}$ for  the underlying coding scheme  to that of MDS codes, in percentage. Intuitively, $G_{\rm cod}$ for a coding scheme determines \emph{how much gain} this scheme attains and $g_{opt}$ indicates  \emph{how close} this scheme is to the optimal solution. 
 {Observe that RM-coded computation with the optimal MAP decoder is able to achieve very close to the performance of optimal MDS-coded computation, and also outperform binary RC for relatively small number of servers. Our low-complexity decoding algorithm is also able to achieve the same performance as that of the MAP decoder for the number of servers $n=8$, $16$, $32$, and $64$ where the optimal values of $k^*$ correspond to RM codes. For the larger values of $n$, the optimal $k^*$'s correspond to RM sub-codes and one needs to extend our algorithm to, possibly, achieve close to the performance of the optimal MAP decoder with a low complexity.}
 Figure \ref{F1} shows that random linear codes have weak performance in the beginning  but they quickly approach the optimal $T_{\rm avg}$ so that they have small gaps to the optimal values, e.g., $g_{\rm opt}=0.5\%$ for $n=512$. Also, observe that RM codes significantly outperform polar codes.
  
  In the case of $\mu=1$, by numerically solving \eqref{RR}, we have for the asymptotically-optimal encoding rate $R^*=0.6822$. 
 Motivated by this fact, in Figure \ref{F2}, the rate of all discussed underlying coding schemes is fixed to $R^*$, and  $nT_{\rm avg}$ is plotted for large blocklengths. Therefore, $T_{\rm avg}$ is not optimized over  rates for the results demonstrated in this figure. Additionally,  the polar code is  designed with $\epsilon_d=1-R^*=0.3178$, which makes the code to be capacity-achieving for an erasure channel with capacity equal to $R^*$. 
 Furthermore, Figure \ref{F2} suggests that RM codes approach very close to the optimal performance, and also do so relatively fast.
 \begin{figure}
	\centering
	\includegraphics[trim={0cm 0.2cm 0cm 0cm},width=3.6in]{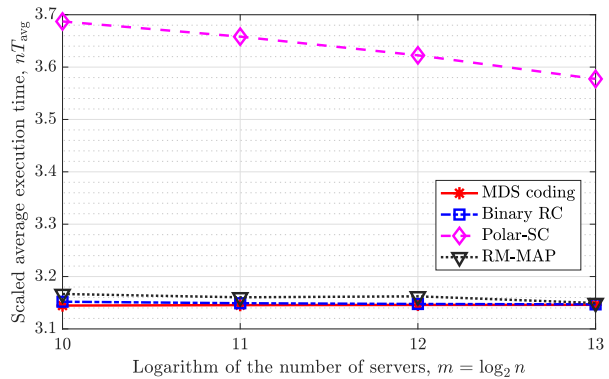}
	\caption{Scaled average execution time of a homogeneous distributed computing system with $\mu=1$ using various coding schemes for asymptotically large number of servers $n=1024$, $2048$, $4096$, and $8192$.}
	\label{F2}
\end{figure}
 
{In order to further demonstrate the applicability of the proposed approach, we analyzed the average execution time of coded distributed computing systems under shifted Weibull distribution which generalizes the shifted exponential distribution considered in the paper. Specifically, there is a parameter $\alpha$ associated with the Weibull distribution, as specified in Appendix\,\ref{AppA}, and for $\alpha = 1$, the (shifted) Weibull distribution is simplified to (shifted) exponential. This in turn provides a higher flexibility in statistical modeling \cite{coles2001introduction} such as the run-time of the computational servers \cite{reisizadeh2017coded}. Using \Lref{lemma_ap1} in Appendix \ref{AppA}, we computed the average execution time of various coding schemes such as MDS, polar with SC, RM with MAP, and RM with our low-complexity decoder. As shown in Figure \ref{F4} and also summarized in Table II, our low-complexity decoder is able to achieve very close to the performance of MDS-coded computing even for very small number of servers.}

\begin{table}[t]
     		\centering
     		\caption{{Average execution time and optimal $k^*$ values for different coding schemes under shifted Weibull distribution with parameters $\mu=1$ and $\alpha=2$.}}
     		\label{table}\vspace{-0.2cm}
   		\begin{tabular}{||M{0.08in}||M{0.50in}|M{0.5in}|M{0.5in}|M{0.5in}||}
     			\hline\hline
     \vspace{0.1cm}	$n$ & 
     {MDS}&
     {Polar-SC}&
    {RM-MAP}&
    {RM with our decoder}\\
     			\cline{2-5}
    &$(k^*,T_{\rm avg})$ & 
     $(k^*,T_{\rm avg})$ & 
    $(k^*,T_{\rm avg})$ & 
    $(k^*,T_{\rm avg})$
    \\ \hline	\hline	
$8$ & $(7,0.3163)$ & $(7,0.3247)$ & $(7,0.3163)$ & $(7,0.3163)$ \\\hline
$16$ & $(14,0.1633)$ & $(15,0.1676)$ & $(15,0.1637)$ & $(15,0.1637)$ \\\hline
$32$ & $(28,0.0832)$ & $(31,0.0878)$ & $(26,0.0856)$ & $(26,0.0857)$ \\\hline
$64$ & $(55,0.0420)$ & $(58,0.0446)$ & $(57,0.0425)$ & $(57,0.0426)$ \\\hline\hline
\end{tabular}
\end{table}

\begin{figure}
	\centering
	\includegraphics[trim={0cm 0.2cm 0cm 0cm},width=3.6in]{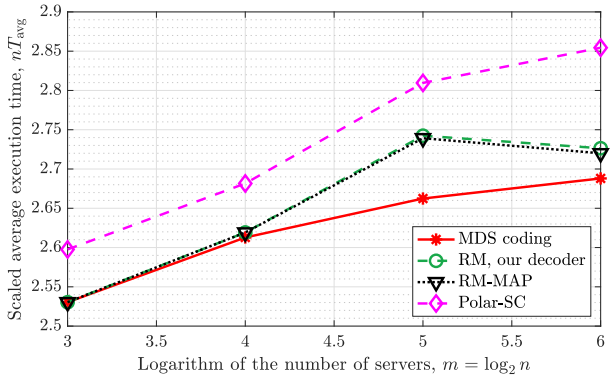}
	\caption{Scaled average execution time of various schemes under shifted Weibull distribution with parameters $\mu=1$ and $\alpha=2$.}
	\label{F4}
\end{figure}

\section{Conclusions}\label{conclusions}
In this paper, we presented a coding-theoretic approach toward coded distributed computing systems by connecting the problem of characterizing their average execution time to the traditional problem of finding the error probability of a coding scheme over erasure channels. Using this connection, we provided results on the performance of coded distributed computing systems, such as their best performance bounds and asymptotic results using binary random linear codes. 
Accordingly, we established the existence of \textit{good} binary linear codes that attain (asymptotically) the best performance any linear code can achieve while maintaining numerical stability against the inevitable rounding errors in practice.
To enable a practical approach for coded computation, we developed a low-complexity algorithm for decoding RM codes over erasure channels, involving only additions and subtractions {(and inverting small matrices of size less than $\log n+1$)}. Our RM-coded computation scheme not only has close-to-optimal performance and explicit construction but also works over real-valued data inputs with a low-complexity decoding.

{
\appendices
\section{Average Execution Time under Shifted Weibull Distribution}\label{AppA}
\begin{lemma}\label{lemma_ap1}
	The average execution time of a coded distributed computing system, using a given $(n,k)$ linear code and assuming shifted Weibull distribution on the run-time of individual nodes, can be characterized as
\begin{align}\label{avgtime_weibull}
T_{\rm avg}=\frac{1}{k}+\frac{1}{\mu k \alpha}\int_{0}^{1}\frac{P_e(\epsilon,n)}{\epsilon\left[\ln(1/\epsilon)\right]^{1-1/\alpha}}d\epsilon,
\end{align}
where $\alpha>0$ represents the shape parameter.
\end{lemma}
\begin{proof}
The CDF of the run time $T_i$ of each local node $i$ under shifted Weibull distribution and the homogeneous system model of this paper can be expressed as \cite{coles2001introduction}:
\begin{align}\label{CDF_weibull}
\Pr(T_i\leq t)=1-\exp\left(-[\mu(kt-1)]^{\alpha}\right),~~~\forall ~\!t\geq 1/k.
\end{align}
Equivalently, the erasure probability $\epsilon(t)$ is equal to one for $t<1/k$ and, otherwise, is equal to $\exp\left(-[\mu(kt-1)]^{\alpha}\right)$ for $t\geq 1/k$.
Now using \eqref{ET1} and given that for the shifted Weibull distribution $d\epsilon(\tau)/d\tau=-\mu k\alpha \epsilon(\tau)\left[\ln(1/\epsilon(\tau))\right]^{1-1/\alpha}$, and that $P_e(\epsilon(\tau),n)=1$ for all $\tau\leq1/k$, we have \eqref{avgtime_weibull} by the change of variables. Note that the result simplifies to the case of shifted exponential distribution for $\alpha=1$, presented in \Lref{lemma1}.
\end{proof}

}

\bibliographystyle{IEEEtran}
\bibliographystyle{IEEEtran}

\bibliography{IEEEabrv}

\begin{thebibliography}{10}
\providecommand{\url}[1]{#1}
\csname url@samestyle\endcsname
\providecommand{\newblock}{\relax}
\providecommand{\bibinfo}[2]{#2}
\providecommand{\BIBentrySTDinterwordspacing}{\spaceskip=0pt\relax}
\providecommand{\BIBentryALTinterwordstretchfactor}{4}
\providecommand{\BIBentryALTinterwordspacing}{\spaceskip=\fontdimen2\font plus
\BIBentryALTinterwordstretchfactor\fontdimen3\font minus
  \fontdimen4\font\relax}
\providecommand{\BIBforeignlanguage}[2]{{%
\expandafter\ifx\csname l@#1\endcsname\relax
\typeout{** WARNING: IEEEtran.bst: No hyphenation pattern has been}%
\typeout{** loaded for the language `#1'. Using the pattern for}%
\typeout{** the default language instead.}%
\else
\language=\csname l@#1\endcsname
\fi
#2}}
\providecommand{\BIBdecl}{\relax}
\BIBdecl

\bibitem{jamali2019coded}
M.~V. Jamali, M.~Soleymani, and H.~Mahdavifar, ``Coded distributed computing:
  Performance limits and code designs,'' in \emph{2019 IEEE Information Theory
  Workshop (ITW)}, pp. 1--5.

\bibitem{ananthanarayanan2013effective}
G.~Ananthanarayanan, A.~Ghodsi, S.~Shenker, and I.~Stoica, ``Effective
  straggler mitigation: Attack of the clones,'' in \emph{Presented as part of
  the 10th {USENIX} Symposium on Networked Systems Design and Implementation
  ({NSDI} 13)}, 2013, pp. 185--198.

\bibitem{vulimiri2013low}
A.~Vulimiri, P.~B. Godfrey, R.~Mittal, J.~Sherry, S.~Ratnasamy, and S.~Shenker,
  ``Low latency via redundancy,'' in \emph{Proceedings of the ninth ACM
  conference on Emerging networking experiments and technologies}.\hskip 1em
  plus 0.5em minus 0.4em\relax ACM, 2013, pp. 283--294.

\bibitem{gardner2015reducing}
K.~Gardner, S.~Zbarsky, S.~Doroudi, M.~Harchol-Balter, and E.~Hyytia,
  ``Reducing latency via redundant requests: Exact analysis,'' \emph{ACM
  SIGMETRICS Performance Evaluation Review}, vol.~43, no.~1, pp. 347--360,
  2015.

\bibitem{jonas2017occupy}
E.~Jonas, Q.~Pu, S.~Venkataraman, I.~Stoica, and B.~Recht, ``Occupy the cloud:
  Distributed computing for the 99\%,'' in \emph{Proceedings of the 2017
  Symposium on Cloud Computing}.\hskip 1em plus 0.5em minus 0.4em\relax ACM,
  2017, pp. 445--451.

\bibitem{lee2018speeding}
K.~Lee, M.~Lam, R.~Pedarsani, D.~Papailiopoulos, and K.~Ramchandran, ``Speeding
  up distributed machine learning using codes,'' \emph{IEEE Trans. Inf.
  Theory}, vol.~64, no.~3, pp. 1514--1529, 2018.

\bibitem{li2020coded}
S.~Li and S.~Avestimehr, ``Coded computing: Mitigating fundamental bottlenecks
  in large-scale distributed computing and machine learning,'' \emph{Found.
  Trends Commun. Inf. Theory, vol. 17, no. 1, pp. 1–148.}, 2020.

\bibitem{li2016unified}
S.~Li, M.~A. Maddah-Ali, and A.~S. Avestimehr, ``A unified coding framework for
  distributed computing with straggling servers,'' in \emph{2016 IEEE Globecom
  Workshops (GC Wkshps)}, pp. 1--6.

\bibitem{reisizadeh2017coded}
A.~Reisizadeh, S.~Prakash, R.~Pedarsani, and A.~S. Avestimehr, ``Coded
  computation over heterogeneous clusters,'' \emph{arXiv preprint
  arXiv:1701.05973}, 2017.

\bibitem{li2017coding}
S.~Li, M.~A. Maddah-Ali, and A.~S. Avestimehr, ``Coding for distributed fog
  computing,'' \emph{IEEE Commun. Mag.}, vol.~55, no.~4, pp. 34--40, 2017.

\bibitem{yang2017computing}
Y.~Yang, P.~Grover, and S.~Kar, ``Computing linear transformations with
  unreliable components,'' \emph{IEEE Trans. Inf. Theory}, vol.~63, no.~6, pp.
  3729--3756, 2017.

\bibitem{lee2017high}
K.~Lee, C.~Suh, and K.~Ramchandran, ``High-dimensional coded matrix
  multiplication,'' in \emph{IEEE Int. Symp. Inf. Theory (ISIT)}.\hskip 1em
  plus 0.5em minus 0.4em\relax IEEE, 2017, pp. 2418--2422.

\bibitem{dutta2016short}
S.~Dutta, V.~Cadambe, and P.~Grover, ``Short-dot: Computing large linear
  transforms distributedly using coded short dot products,'' in \emph{Adv. in
  Neural Info. Proc. Systems (NIPS)}, 2016, pp. 2100--2108.

\bibitem{wang2018coded}
S.~Wang, J.~Liu, and N.~Shroff, ``Coded sparse matrix multiplication,''
  \emph{arXiv preprint arXiv:1802.03430}, 2018.

\bibitem{dutta2019optimal}
S.~Dutta, M.~Fahim, F.~Haddadpour, H.~Jeong, V.~Cadambe, and P.~Grover, ``On
  the optimal recovery threshold of coded matrix multiplication,'' \emph{IEEE
  Trans. Inf. Theory}, vol.~66, no.~1, pp. 278--301, 2019.

\bibitem{yu2020straggler}
Q.~Yu, M.~A. Maddah-Ali, and A.~S. Avestimehr, ``Straggler mitigation in
  distributed matrix multiplication: Fundamental limits and optimal coding,''
  \emph{IEEE Trans. Inf. Theory}, vol.~66, no.~3, pp. 1920--1933, 2020.

\bibitem{prakash2020coded}
S.~Prakash, S.~Dhakal, M.~R. Akdeniz, Y.~Yona, S.~Talwar, S.~Avestimehr, and
  N.~Himayat, ``Coded computing for low-latency federated learning over
  wireless edge networks,'' \emph{IEEE Journal on Selected Areas in
  Communications}, vol.~39, no.~1, pp. 233--250, 2020.

\bibitem{jahani2018codedsketch}
T.~Jahani-Nezhad and M.~A. Maddah-Ali, ``Codedsketch: A coding scheme for
  distributed computation of approximated matrix multiplication,'' \emph{arXiv
  preprint arXiv:1812.10460}, 2018.

\bibitem{higham2002accuracy}
N.~J. Higham, \emph{Accuracy and stability of numerical algorithms}.\hskip 1em
  plus 0.5em minus 0.4em\relax Siam, 2002, vol.~80.

\bibitem{yu2019lagrange}
Q.~Yu, S.~Li, N.~Raviv, S.~M.~M. Kalan, M.~Soltanolkotabi, and S.~A.
  Avestimehr, ``Lagrange coded computing: Optimal design for resiliency,
  security, and privacy,'' in \emph{The 22nd International Conference on
  Artificial Intelligence and Statistics}.\hskip 1em plus 0.5em minus
  0.4em\relax PMLR, 2019, pp. 1215--1225.

\bibitem{reisizadeh2019coded}
A.~Reisizadeh, S.~Prakash, R.~Pedarsani, and A.~S. Avestimehr, ``Coded
  computation over heterogeneous clusters,'' \emph{IEEE Trans. Inf. Theory},
  vol.~65, no.~7, pp. 4227--4242, 2019.

\bibitem{li2016coded}
S.~Li, M.~A. Maddah-Ali, and A.~S. Avestimehr, ``Coded distributed computing:
  Straggling servers and multistage dataflows,'' in \emph{54th Annual Allerton
  Conference}.\hskip 1em plus 0.5em minus 0.4em\relax IEEE, 2016, pp. 164--171.

\bibitem{yu2018lagrange}
Q.~Yu, S.~Li, N.~Raviv, S.~M.~M. Kalan, M.~Soltanolkotabi, and S.~A.
  Avestimehr, ``Lagrange coded computing: Optimal design for resiliency,
  security, and privacy,'' in \emph{The 22nd International Conference on
  Artificial Intelligence and Statistics}, 2019, pp. 1215--1225.

\bibitem{yu2020entangled}
Q.~Yu and A.~S. Avestimehr, ``Entangled polynomial codes for secure, private,
  and batch distributed matrix multiplication: Breaking the "cubic" barrier,''
  in \emph{2020 IEEE International Symposium on Information Theory
  (ISIT)}.\hskip 1em plus 0.5em minus 0.4em\relax IEEE, 2020, pp. 245--250.

\bibitem{aliasgari2020private}
M.~Aliasgari, O.~Simeone, and J.~Kliewer, ``Private and secure distributed
  matrix multiplication with flexible communication load,'' \emph{IEEE Trans.
  on Information Forensics and Security}, vol.~15, pp. 2722--2734, 2020.

\bibitem{d2020gasp}
R.~G. D’Oliveira, S.~El~Rouayheb, and D.~Karpuk, ``{GASP} codes for secure
  distributed matrix multiplication,'' \emph{IEEE Trans. Inf. Theory}, vol.~66,
  pp. 4038--4050, 2020.

\bibitem{bitar2019private}
R.~Bitar, Y.~Xing, Y.~Keshtkarjahromi, V.~Dasari, S.~El~Rouayheb, and
  H.~Seferoglu, ``Private and rateless adaptive coded matrix-vector
  multiplication,'' \emph{EURASIP Journal on Wireless Communications and
  Networking}, vol. 2021, no.~1, pp. 1--25, 2021.

\bibitem{nodehi2019secure}
H.~A. Nodehi and M.~A. Maddah-Ali, ``Secure coded multi-party computation for
  massive matrix operations,'' \emph{IEEE Transactions on Information Theory},
  2021.

\bibitem{soleymani2021list}
M.~Soleymani, R.~E. Ali, H.~Mahdavifar, and A.~S. Avestimehr, ``List-decodable
  coded computing: Breaking the adversarial toleration barrier,'' \emph{arXiv
  preprint arXiv:2101.11653}, 2021.

\bibitem{soleymani2020privacy}
M.~Soleymani, H.~Mahdavifar, and A.~S. Avestimehr, ``Privacy-preserving
  distributed learning in the analog domain,'' \emph{arXiv
  preprint:2007.08803}, 2020.

\bibitem{soleymani2020analog}
------, ``Analog {Lagrange} coded computing,'' \emph{IEEE Journal on Selected
  Areas in Information Theory (JSAIT): Special issue on Privacy and Security of
  Information Systems}, 2021.

\bibitem{jahani2020berrut}
T.~Jahani-Nezhad and M.~A. Maddah-Ali, ``Berrut approximated coded computing:
  Straggler resistance beyond polynomial computing,'' \emph{arXiv preprint
  arXiv:2009.08327}, 2020.

\bibitem{roth2020analog}
R.~M. Roth, ``Analog error-correcting codes,'' \emph{IEEE Transactions on
  Information Theory}, vol.~66, no.~7, pp. 4075--4088, 2020.

\bibitem{soleymani2019analog}
M.~Soleymani and H.~Mahdavifar, ``Analog subspace coding: A new approach to
  coding for non-coherent wireless networks,'' in \emph{2020 IEEE International
  Symposium on Information Theory (ISIT)}.\hskip 1em plus 0.5em minus
  0.4em\relax IEEE, 2020, pp. 31--36.

\bibitem{bartan2019polar}
B.~Bartan and M.~Pilanci, ``Polar coded distributed matrix multiplication,''
  \emph{arXiv preprint arXiv:1901.06811}, 2019.

\bibitem{kudekar2017reed}
S.~Kudekar, S.~Kumar, M.~Mondelli, H.~D. Pfister, E.~{\c{S}}a{\c{s}}oǧlu, and
  R.~L. Urbanke, ``{Reed-Muller} codes achieve capacity on erasure channels,''
  \emph{IEEE Trans. Inf. Theory}, vol.~63, no.~7, pp. 4298--4316, 2017.

\bibitem{abbe2015reed}
E.~Abbe, A.~Shpilka, and A.~Wigderson, ``{Reed-Muller} codes for random
  erasures and errors,'' \emph{IEEE Trans. Inf. Theory}, vol.~61, no.~10, pp.
  5229--5252, 2015.

\bibitem{hassani2018almost}
H.~Hassani, S.~Kudekar, O.~Ordentlich, Y.~Polyanskiy, and R.~Urbanke, ``Almost
  optimal scaling of {Reed-Muller} codes on {BEC} and {BSC} channels,'' in
  \emph{IEEE Int. Symp. Inf. Theory}.\hskip 1em plus 0.5em minus 0.4em\relax
  IEEE, 2018, pp. 311--315.

\bibitem{ye2020recursive}
M.~Ye and E.~Abbe, ``Recursive projection-aggregation decoding of {Reed-Muller}
  codes,'' \emph{IEEE Trans. Inf. Theory}, vol.~66, no.~8, pp. 4948--4965,
  2020.

\bibitem{reed1954class}
I.~Reed, ``A class of multiple-error-correcting codes and the decoding
  scheme,'' \emph{Transactions of the IRE Professional Group on Information
  Theory}, vol.~4, no.~4, pp. 38--49, 1954.

\bibitem{dumer2006soft}
I.~Dumer and K.~Shabunov, ``Soft-decision decoding of {Reed-Muller} codes:
  recursive lists,'' \emph{IEEE Trans. Inf. Theory}, vol.~52, no.~3, pp.
  1260--1266, 2006.

\bibitem{saptharishi2017efficiently}
R.~Saptharishi, A.~Shpilka, and B.~L. Volk, ``Efficiently decoding
  {Reed-Muller} codes from random errors,'' \emph{IEEE Trans. Inf. Theory},
  vol.~63, no.~4, pp. 1954--1960, 2017.

\bibitem{santi2018decoding}
E.~Santi, C.~Hager, and H.~D. Pfister, ``Decoding {Reed-Muller} codes using
  minimum-weight parity checks,'' in \emph{2018 IEEE Int. Symp. Inf. Theory
  (ISIT)}.\hskip 1em plus 0.5em minus 0.4em\relax IEEE, 2018, pp. 1296--1300.

\bibitem{liang2014tofec}
G.~Liang and U.~C. Kozat, ``{TOFEC:} achieving optimal throughput-delay
  trade-off of cloud storage using erasure codes,'' in \emph{IEEE Conf.
  Computer Commun. (INFOCOM)}.\hskip 1em plus 0.5em minus 0.4em\relax IEEE,
  2014, pp. 826--834.

\bibitem{kahn1995probability}
J.~Kahn, J.~Koml{\'o}s, and E.~Szemer{\'e}di, ``On the probability that a
  random$\pm$1-matrix is singular,'' \emph{Journal of the American Mathematical
  Society}, vol.~8, no.~1, pp. 223--240, 1995.

\bibitem{macwilliams1977theory}
F.~J. MacWilliams and N.~J.~A. Sloane, \emph{The theory of error-correcting
  codes}.\hskip 1em plus 0.5em minus 0.4em\relax New York: North-Holland, 1977.

\bibitem{roth1989construction}
R.~M. Roth and A.~Lempel, ``A construction of non-{R}eed-{S}olomon type {MDS}
  codes,'' \emph{IEEE transactions on information theory}, vol.~35, no.~3, pp.
  655--657, 1989.

\bibitem{beelen2018explicit}
P.~Beelen and L.~Jin, ``Explicit {MDS} codes with complementary duals,''
  \emph{IEEE Transactions on Information Theory}, vol.~64, no.~11, pp.
  7188--7193, 2018.

\bibitem{chen2017new}
B.~Chen and H.~Liu, ``New constructions of {MDS} codes with complementary
  duals,'' \emph{IEEE Transactions on Information Theory}, vol.~64, no.~8, pp.
  5776--5782, 2017.

\bibitem{carlet2018euclidean}
C.~Carlet, S.~Mesnager, C.~Tang, and Y.~Qi, ``Euclidean and {H}ermitian {LCD}
  {MDS} codes,'' \emph{Designs, Codes and Cryptography}, vol.~86, no.~11, pp.
  2605--2618, 2018.

\bibitem{demmel1997applied}
J.~W. Demmel, \emph{Applied numerical linear algebra}.\hskip 1em plus 0.5em
  minus 0.4em\relax Siam, 1997, vol.~56.

\bibitem{boley1992algorithmic}
D.~L. Boley, R.~P. Brent, G.~H. Golub, and F.~T. Luk, ``Algorithmic fault
  tolerance using the {Lanczos} method,'' \emph{SIAM Journal on Matrix Analysis
  and Applications}, vol.~13, no.~1, pp. 312--332, 1992.

\bibitem{ferreira2000stability}
P.~J. Ferreira, ``Stability issues in error control coding in the complex
  field, interpolation, and frame bounds,'' \emph{IEEE Signal Processing
  Letters}, vol.~7, no.~3, pp. 57--59, 2000.

\bibitem{ferreira2003stable}
P.~J. Ferreira and J.~M. Vieira, ``Stable {DFT} codes and frames,'' \emph{IEEE
  Signal Processing Letters}, vol.~10, no.~2, pp. 50--53, 2003.

\bibitem{henkel1988multiple}
W.~Henkel, ``Multiple error correction with analog codes,'' in
  \emph{International Conference on Applied Algebra, Algebraic Algorithms, and
  Error-Correcting Codes}.\hskip 1em plus 0.5em minus 0.4em\relax Springer,
  1988, pp. 239--249.

\bibitem{marvasti1999efficient}
F.~Marvasti, M.~Hasan, M.~Echhart, and S.~Talebi, ``Efficient algorithms for
  burst error recovery using {FFT} and other transform kernels,'' \emph{IEEE
  Transactions on Signal Processing}, vol.~47, no.~4, pp. 1065--1075, 1999.

\bibitem{chen2005numerically}
Z.~Chen and J.~Dongarra, ``Numerically stable real number codes based on random
  matrices,'' in \emph{International Conference on Computational
  Science}.\hskip 1em plus 0.5em minus 0.4em\relax Springer, 2005, pp.
  115--122.

\bibitem{chen2005condition}
Z.~Chen and J.~J. Dongarra, ``Condition numbers of gaussian random matrices,''
  \emph{SIAM Journal on Matrix Analysis and Applications}, vol.~27, no.~3, pp.
  603--620, 2005.

\bibitem{chen2009optimal}
Z.~Chen, ``Optimal real number codes for fault tolerant matrix operations,'' in
  \emph{Proceedings of the Conference on High Performance Computing Networking,
  Storage and Analysis}, 2009, pp. 1--10.

\bibitem{rudelson2008littlewood}
M.~Rudelson and R.~Vershynin, ``The littlewood--offord problem and
  invertibility of random matrices,'' \emph{Advances in Mathematics}, vol. 218,
  no.~2, pp. 600--633, 2008.

\bibitem{rudelson2010non}
------, ``Non-asymptotic theory of random matrices: extreme singular values,''
  in \emph{Proceedings of the International Congress of Mathematicians 2010
  (ICM 2010) (In 4 Volumes) Vol. I: Plenary Lectures and Ceremonies Vols.
  II--IV: Invited Lectures}.\hskip 1em plus 0.5em minus 0.4em\relax World
  Scientific, 2010, pp. 1576--1602.

\bibitem{bourgain2010singularity}
J.~Bourgain, V.~H. Vu, and P.~M. Wood, ``On the singularity probability of
  discrete random matrices,'' \emph{Journal of Functional Analysis}, vol. 258,
  no.~2, pp. 559--603, 2010.

\bibitem{kaufman2012weight}
T.~Kaufman, S.~Lovett, and E.~Porat, ``Weight distribution and list-decoding
  size of {Reed-Muller} codes,'' \emph{IEEE Trans. Inf. Theory}, vol.~58,
  no.~5, pp. 2689--2696, 2012.

\bibitem{jamali2021reed}
M.~V. Jamali, X.~Liu, A.~V. Makkuva, H.~Mahdavifar, S.~Oh, and P.~Viswanath,
  ``{Reed-Muller} subcodes: Machine learning-aided design of efficient soft
  recursive decoding,'' \emph{arXiv preprint arXiv:2102.01671}, 2021.

\bibitem{Arikan}
E.~Arikan, ``Channel polarization: A method for constructing capacity-achieving
  codes for symmetric binary-input memoryless channels,'' \emph{IEEE Trans.
  Inf. Theory}, vol.~55, no.~7, pp. 3051--3073, 2009.

\bibitem{jamali2018low}
M.~V. Jamali and H.~Mahdavifar, ``A low-complexity recursive approach toward
  code-domain {NOMA} for massive communications,'' in \emph{Proc. IEEE Global
  Commun. Conf. (GLOBECOM), Dec. 2018}, pp. 1--6.

\bibitem{jamali2020massive}
------, ``Massive coded-{NOMA} for low-capacity channels: A low-complexity
  recursive approach,'' \emph{arXiv preprint arXiv:2006.06917}, 2020.

\bibitem{coles2001introduction}
S.~Coles, J.~Bawa, L.~Trenner, and P.~Dorazio, \emph{An introduction to
  statistical modeling of extreme values}.\hskip 1em plus 0.5em minus
  0.4em\relax Springer, 2001, vol. 208.

\end{thebibliography}

\end{document}